\documentclass[a4paper]{article}
\usepackage{enumitem}
\usepackage{calc}

\usepackage{lmodern} 

\usepackage{graphicx} 

\usepackage{amsmath, accents}
\usepackage{amssymb,tikz}
\usepackage{pict2e}
\usepackage{amsthm}
\usepackage{amscd}
\usepackage{mathrsfs}
\usepackage{todonotes}
\usetikzlibrary{calc} 

\usepackage{yfonts}

\usepackage{marvosym}
\usepackage[percent]{overpic}

\newtheorem{theorem}{Theorem} [section]

\newtheorem{lemma}[theorem]{Lemma}

\newtheorem{remark}[theorem]{Remark}
\theoremstyle{definition}

\newcommand{\R}{\mathbb{R}}

\newcommand{\re}{\text{\upshape Re\,}}
\newcommand{\im}{\text{\upshape Im\,}}

\newcommand{\He}{{\rm He}}

\let\oldbibliography\thebibliography
\renewcommand{\thebibliography}[1]{\oldbibliography{#1}
\setlength{\itemsep}{-0.5pt}}

\def\XXint#1#2#3{{\setbox0=\hbox{$#1{#2#3}{\int}$}
\vcenter{\hbox{$#2#3$}}\kern-.5\wd0}}

\usepackage{tikz}
\usetikzlibrary{arrows}
\usetikzlibrary{decorations.pathmorphing}
\usetikzlibrary{decorations.markings}
\usetikzlibrary{patterns}
\usetikzlibrary{automata}
\usetikzlibrary{positioning}
\usepackage{tikz-cd}
\tikzset{->-/.style={decoration={
				markings,
				mark=at position #1 with {\arrow{latex}}},postaction={decorate}}}
	
	\tikzset{-<-/.style={decoration={
				markings,
				mark=at position #1 with {\arrowreversed{latex}}},postaction={decorate}}}

\usetikzlibrary{shapes.misc}\tikzset{cross/.style={cross out, draw, 
         minimum size=2*(#1-\pgflinewidth), 
         inner sep=0pt, outer sep=0pt}}

\usepackage{pgfplots}
	
\allowdisplaybreaks	

\numberwithin{equation}{section}

\def\ds{\displaystyle}
\def\bigO{{\cal O}}

\usepackage[colorlinks=true]{hyperref}
\hypersetup{urlcolor=blue, citecolor=red, linkcolor=blue}

\begin{document}
\title{\vspace*{-1.5cm} On the characteristic polynomial of the eigenvalue moduli of random normal matrices}
\author{Sung-Soo Byun\footnote{School of Mathematics, Korea Institute for Advanced Study, 85 Hoegiro, Dongdaemun-gu, Seoul 02455, Republic of Korea. e-mail: sungsoobyun@kias.re.kr}  \mbox{ and }Christophe Charlier\footnote{Centre for Mathematical Sciences, Lund University, 22100 Lund, Sweden. e-mail:	christophe.charlier@math.lu.se
}}

\maketitle

\begin{abstract}
We study the characteristic polynomial $p_{n}(x)=\prod_{j=1}^{n}(|z_{j}|-x)$ where the $z_{j}$ are drawn from the Mittag-Leffler ensemble, i.e. a two-dimensional determinantal point process which generalizes the Ginibre point process. We obtain precise large $n$ asymptotics for the moment generating function $\mathbb{E}[e^{\frac{u}{\pi} \, \im \ln p_{n}(r)}e^{a \, \re \ln p_{n}(r)}]$, in the case where $r$ is in the bulk, $u \in \mathbb{R}$ and $a \in \mathbb{N}$. This expectation involves an $n \times n$ determinant whose weight is supported on the whole complex plane, is rotation-invariant, and has both jump- and root-type singularities along the circle centered at $0$ of radius $r$. This ``circular" root-type singularity differs from earlier works on Fisher-Hartwig singularities, and surprisingly yields a new kind of ingredient in the asymptotics, the so-called \textit{associated Hermite polynomials}. 
\end{abstract}


\section{Introduction and statement of results}\label{section: introduction}
The Mittag-Leffler ensemble with parameters $b>0$ and $\alpha>-1$ is the joint probability distribution 
\begin{align}\label{def of point process}
\frac{1}{n!Z_{n}} \prod_{1 \leq j < k \leq n} |z_{k} -z_{j}|^{2} \prod_{j=1}^{n}|z_{j}|^{2\alpha}e^{-n |z_{j}|^{2b}}d^{2}z_{j}, \qquad z_{1},\ldots,z_{n} \in \mathbb{C},
\end{align}
where $Z_{n}$ is the normalization constant. This determinantal point process can be realized as the eigenvalues of a random normal matrix $M$ with distribution proportional to $|\det(M)|^{2\alpha}e^{-n \, \mathrm{tr}((MM^{*})^{b})}dM$ \cite{Mehta}. The special case $(b,\alpha)=(1,0)$ also corresponds to the eigenvalue distribution of a Ginibre matrix \cite{Ginibre}, i.e. an $n \times n$ matrix with independent complex Gaussian entries with mean $0$ and variance~$\frac{1}{n}$.

\medskip Consider the characteristic polynomial $p_{n}(x)=\prod_{j=1}^{n}(|z_{j}|-x)$ of the process of the moduli $\{|z_{j}|\}_{j=1}^{n}$. The main result of this work is a precise asymptotic formula as $n \to + \infty$ for
\begin{align}\label{moment generating function}
\mathbb{E}\Big[ e^{\frac{u}{\pi} \, \im \ln p_{n}(r)}e^{a \, \re \ln p_{n}(r)} \Big],
\end{align}
where $u \in \mathbb{R}$, $a \in \mathbb{N}:=\{0,1,\ldots\}$, $r \in (0,b^{-\frac{1}{2b}})$ are fixed, and $\ln p_{n}(r) := \ln |p_{n}(r)| + \pi i \, \#\{z_{j}:~|z_{j}|<~r\}$. The macroscopic large $n$ behavior of the $|z_{j}|$ is described by the probability measure $d\mu(y)=2b^{2}y^{2b-1}dy$, whose support is $[0,b^{-\smash{\frac{1}{2b}}}]$ \cite{SaTo}; thus the condition that $r \in (0,b^{-\frac{1}{2b}})$ is fixed means that we focus on ``the bulk regime". By definition, the expectation \eqref{moment generating function} is equal to $D_{n}/Z_{n}$, where
\begin{align}\label{def of Dn intro}
D_{n} & :=\frac{1}{n!} \int_{\mathbb{C}}\ldots \int_{\mathbb{C}} \prod_{1 \leq j < k \leq n} |z_{k} -z_{j}|^{2} \prod_{j=1}^{n}w(z_{j})d^{2}z_{j} = \det \left( \int_{\mathbb{C}} z^{j} \overline{z}^{k} w(z) d^{2}z \right)_{j,k=0}^{n-1},
\end{align}
and the weight $w$ is given by
\begin{align}\label{def of w and omega}
w(z):=|z|^{2\alpha} e^{-n |z|^{2b}} \omega(|z|), \qquad \omega(x) := |x-r|^{a} \begin{cases}
e^{u}, & \mbox{if } x < r, \\
1, & \mbox{if } x \geq r.
\end{cases}
\end{align}
Hence our results can also be seen as large $n$ asymptotics for $n \times n$ determinants whose weight is supported on $\mathbb{C}$, rotation-invariant, has both jump- and root-type singularities along the circle centered at $0$ of radius $r$ (which we will call ``circular" jump- and root-type singularities), and a ``pointwise" root-type singularity at $0$.

\medskip Over the past 50 years or so a lot of works have been done on structured determinants with singularities, and we briefly pause here to review the literature. In their pioneering work \cite{FisherHartwig}, Fisher and Hartwig made a conjecture for the asymptotics of large Toeplitz determinants when the weight is supported on the unit circle and has root- and jump-type singularities---such singularities are now called Fisher-Hartwig singularities. Many authors have contributed in proving this conjecture for certain parameter ranges, among which Lenard \cite{Lenard}, Widom \cite{Widom2}, Basor \cite{Basor}, B\"ottcher and Silbermann \cite{BS1985}, and Ehrhardt \cite{Ehr2001}. A counterexample to the Fisher-Hartwig conjecture was found by Basor and Tracy in \cite{BasTra}, and the corrected conjecture was solved for general values of the parameters by Deift, Its and Krasovsky \cite{DIK}. The study of these singular determinants was motivated mainly from questions arising in the Ising model and impenetrable bosons, see \cite{BasMor, DIKreview} for more historical background. In recent years, these determinants have also attracted considerable attention in the random matrix community. One reason for that is the well-known work \cite{KeaSna} of Keating and Snaith, where numerical evidences were found of links between the characteristic polynomials of unitary and Hermitian random matrices and the zeros of the Riemann zeta function on the critical line. Expectations of powers of the absolute value of the characteristic polynomial of the Gaussian Unitary Ensemble, which are Hankel determinants with a Gaussian weight on $\mathbb{R}$ and root-type singularities, were investigated in \cite{BreHik} and their asymptotics were obtained by Garoni \cite{Garoni} for integer values of the parameters and by Krasovsky \cite{Krasovsky} for the general case. This result was then generalized by Berestycki, Webb and Wong \cite{BerWebbWong} for one-cut regular ensembles. In a different direction, Its and Krasovsky in \cite{ItsKrasovsky} obtained asymptotics of Hankel determinants with a jump-type singularity and a Gaussian weight. Such determinants provide information about the imaginary part of the log-characteristic polynomial of the Gaussian Unitary Ensemble. The results \cite{ItsKrasovsky, BerWebbWong} have been generalized in \cite{Charlier} for general Fisher-Hartwig singularities and one-cut regular ensembles. The case of one-cut regular ensembles with hard edges was then treated in \cite{CharlierGharakhloo}, and the multi-cut case in \cite{CFWW2021}. Strong results on Toeplitz determinants with merging Fisher-Hartwig singularities are also available in the literature \cite{CK2015, Fahs}; these results have been useful to prove a conjecture of \cite{FyoKea} on ``the moments of the moments" of the characteristic polynomial of random unitary matrices, and \cite{CK2015} has also been used by Webb in \cite{Webb} to establish a connection between random matrix theory and Gaussian multiplicative chaos. There exists also a vast literature on other structured determinants with Fisher-Hartwig singularities, see e.g. \cite{DXZ2022, DXZ2022 bis} for Fredholm determinants, \cite{BasorEhrhardt1, ForKea, CGMY2020} for Toeplitz+Hankel determinants, and \cite{CharlierMB} for a biorthogonal generalization of Hankel determinants.

\medskip The literature on determinants associated with a singular weight supported on $\mathbb{C}$ is more limited. For $a=0$, \eqref{moment generating function} is the moment generating function of the disk counting function
\begin{align}\label{disk counting fct}
\mathbb{E}\Big[ e^{\frac{u}{\pi} \, \im \ln p_{n}(r)} \Big] = \mathbb{E}\Big[ e^{u \, \#\{z_{j}: |z_{j}|< r\}} \Big],
\end{align}
and in this case $w$ is discontinuous along a circle but has no ``circular" root-type singularity. Counting statistics of two-dimensional point processes have attracted a lot of interest in recent years \cite{LeeRiser2016, CE2020, LMS2018, L et al 2019, FenzlLambert, ES2020, SDMS2020, SDMS2021, Charlier 2d jumps, AkemannSungsoo, CLmerging}. The first two terms in the large $n$ asymptotics of \eqref{disk counting fct} were derived in \cite{CE2020, L et al 2019}\footnote{\cite{CE2020} considers counting statistics on more general domains (not only centered disks) for a class of determinantal processes on a K\"{a}hler manifold (which includes Ginibre), and \cite{L et al 2019} considers general ``one-cut" rotation-invariant potentials (including Mittag-Leffler).}. More precise asymptotics, including the third term of order $1$, were obtained in \cite{FenzlLambert} for the Ginibre ensemble (i.e. $(b,\alpha)=(1,0)$)\footnote{Some generalizations of the Ginibre point process  (different from the Mittag-Leffler ensemble) and some hyperbolic models have also been considered in \cite{FenzlLambert}.}, and the asymptotics of \eqref{disk counting fct} including the fourth term of order $n^{-\frac{1}{2}}$ were then obtained in \cite{Charlier 2d jumps} for general $b>0$ and $\alpha>-1$. Several works on determinants with ``pointwise" root-type singularities in dimension two are also available in the literature. In \cite{BBLM2015, BGM17, LeeYang, BEG18, LeeYang2, LeeYang3}, the orthogonal polynomials for the planar Gaussian weight perturbed with a finite number of ``pointwise" root-type singularities have been studied, see also \cite{AKS2018} where microscopic properties of the associated point process have been analyzed.  Building on \cite{BBLM2015, LeeYang}, Webb and Wong in \cite{WebbWong} obtained a precise asymptotic formula for $\mathbb{E} [ e^{a \, \re \ln q_{n}(r)} ]$ where $a\in \mathbb{C}$, $\re a>-2$, $r<1$ and $q_{n}$ is the characteristic polynomial of a Ginibre matrix, i.e. $q_{n}(x) = \prod_{j=1}^{n}(z_{j}-x)$ and the $z_{j}$ are drawn from \eqref{def of point process} with $(b,\alpha)=(1,0)$. This expectation involves a determinant with a single  ``pointwise" root-type singularity in the bulk. The case where $r$ is close to $1$, which corresponds to the edge regime, was then investigated by Dea\~{n}o and Simm in \cite{DeanoSimm}. Determinants with two merging planar ``pointwise" root-type singularities were also considered in \cite{DeanoSimm}, and the asymptotics were found to involve some Painlev\'{e} transcendents. 

\medskip Determinants with ``circular" root-type singularities have not been considered before to our knowledge. A main difficulty in the analysis of ``pointwise" root-type singularities in dimension two stems from the fact that they break the rotation-invariance of the weight (unless of course if they are located at $0$). The ``circular" root-type singularities preserve the rotation-invariance of the weight, which makes them simpler to analyze in this respect, but they also pose a series of new challenges, which we discuss at the end of this section, and which we have been able to overcome only for integer values of $a$. Interestingly, these ``circular" root-type singularities also produce some \textit{associated Hermite polynomials} (see below for the definition) in the asymptotics. This came as a surprise to us and appears to be completely new. There are of course many \textit{exact} formulas in random matrix theory which involve the Hermite polynomials, but we are not aware of an earlier work where these polynomials show up explicitly in the \textit{asymptotics} of large determinants (let alone the associated Hermite polynomials). For comparison, ``pointwise" root-type singularities typically produce other kinds of ingredients in the asymptotics, such as Barnes' $G$-function (as was discovered by Basor \cite{Basor} in dimension one and by Webb and Wong \cite{WebbWong} in dimension two), and ``circular" jump-type singularities involve the error function. We also mention that ensembles with ``circular" root-type singularities have been studied in \cite{ZS2000, Seo}, and ensembles with ``elliptic" root-type singularities in \cite{NAKP2020}. In \cite{ZS2000, NAKP2020, Seo}, the singularities are located at the hard edge and the focus was on the leading order behavior of the kernel; in particular, the (associated) Hermite polynomials do not show up in these works.

\medskip The $\nu$-th \emph{associated} Hermite polynomials $\{ \He_k^{(\nu)}:k=0,1,\dots \}$ are defined recursively by
\begin{align}\label{Hermite recurrence}
\begin{cases}
\He_{k+1}^{(\nu)}(x)= x \, \He_k^{(\nu)}(x)-(k+\nu)\He^{(\nu)}_{k-1}(x), & k \geq 1, \\[0.1cm]
\He_0^{(\nu)}(x)=1, \quad \He_1^{(\nu)}(x)=x,
\end{cases}
\end{align}
and satisfy the orthogonality relations \cite{AW1984}
\begin{align}\label{orthogonality associated Hermite}
\int_{-\infty}^{+\infty} \He_{k}^{(\nu)}(x) \He_{\ell}^{(\nu)}(x) \frac{dx}{|D_{-\nu}(ix)|^{2}} = \sqrt{2\pi} \, (k+\nu)! \, \delta_{k\ell}, \qquad k,\ell = 0,1,\ldots
\end{align}
where $D_{-\nu}$ is the parabolic cylinder function. These polynomials are explicitly given by 
\begin{align*}
\He_k^{(\nu)}(x) = \begin{cases}
\ds k!   \sum_{\ell=0 }^{ \lfloor k/2 \rfloor  } \frac{(-1)^\ell}{\ell!(k-2\ell)!}\frac{x^{k-2\ell}}{2^\ell} = \Big[ \frac{d}{dt} \Big]^k \Big[ e^{xt-\frac{t^2}{2}}\Big] \Big|_{t=0}, & \mbox{if } \nu=0, \\[0.4cm]
\ds \sum_{\ell=0}^{\lfloor k/2\rfloor} \frac{(-1)^{\ell}}{(k-2\ell)!}\bigg( \sum_{j=0}^{\ell} \frac{(k-j)!(\nu-1+j)!}{j!(\ell-j)!(\nu-1)!2^{\ell-j}} \bigg)x^{k-2\ell}, & \mbox{if } \nu \geq 1,
\end{cases}
\end{align*}
see \cite[eq. (4.12)]{Wunsche}. For $\nu=0$, they reduce to the standard Hermite polynomials\footnote{There are two commonly used Hermite polynomials in the literature, denoted $\He_{k}$ and $H_{k}$, and which are related by $H_{k}(x) = 2^{\smash{\frac{k}{2}}}\He_{n}(\sqrt{2}\, x)$. For us it is more convenient to work with $\He_{k}$.}, i.e. $\He_{k}^{(0)}(x) = \He_{k}(x)$ for all $k \in \mathbb{N}$. We refer to \cite{Van Assche} for basic properties of general associated polynomials, and to \cite{Wunsche} for a focus on the Hermite case. 

\medskip Only the polynomials $\{\He_{k},\He_{k}^{(1)}:k=0,1,\ldots\}$, corresponding to $\nu=0$ and $\nu=1$, will appear in our asymptotic formula. Our results can be presented in a unified way if we formally define $\He_{k}$, $\He_{k}^{(1)}$ for the first few negative $k$ as follows:
\begin{align}
& \He_{-1}(x) := 0, & & \He_{-2}(x) := 1, & & \He_{-3}(x) := -\frac{x}{2}, \label{negative Hermite polynomials} \\
& \He_{-1}^{(1)}(x) := 0, & & \big[ k \He_{k-2}^{(1)}(x) \big]_{k=0} := -1, & & \big[ k \He_{k-3}^{(1)}(x) \big]_{k=0} := x, & & \big[ k \He_{k-4}^{(1)}(x) \big]_{k=0} := -\frac{x^{2}+1}{2}. \nonumber
\end{align}
These definitions are consistent with the recurrence \eqref{Hermite recurrence}. For general $a \in \mathbb{N}$, we define 
\begin{align}
& p_{0,a}(x):= \frac1{i^a} \He_a(ix) =   \sum_{s=0 }^{ \lfloor a/2 \rfloor  } \frac{ a! }{s!(a-2s)!}   \frac{x^{a-2s}}{2^s}, \label{def of p0a} 
 \\ 
& q_{0,a}(x):= \frac1{i^{a-1}} \He_{a-1}^{(1)}(ix) =  \sum_{s=0 }^{ \lfloor (a-1)/2 \rfloor  }  \bigg( \sum_{j=0}^s \frac{(a-1-j)!\,2^j}{(s-j)! } \bigg) \frac1{(a-1-2s)!} \frac{ x^{a-1-2s}}{2^s}, \label{def of q0ap1} \\
& p_{1,a}(x)  := -\frac{a}{2} p_{0,a+1}(x) -ab  \Big(  p_{0,a+1}(x)- (3a-1) \,p_{0,a-1}(x) + \frac{5}{3}(a-1)(a-2) \, p_{0,a-3}(x) \Big), \label{def of p1a} \\
& q_{1,a}(x) := -\frac{a}{2} q_{0,a+1}(x) -b  \Big(  aq_{0,a+1}(x)- (3a-1) [a\,q_{0,a-1}(x)] + \frac{5}{3} [a(a-1)(a-2) \, q_{0,a-3}(x)] \Big), \label{def of q1a}
\end{align}
where the brackets in \eqref{def of q1a} emphasize that for the first values of $a$, one needs to use \eqref{negative Hermite polynomials}, namely
\begin{align}
& [a q_{0,a-1}(x)] := \begin{cases}
1, & \mbox{if } a=0, \\
a q_{0,a-1}(x), & \mbox{if } a \geq 1,
\end{cases} \nonumber \\
& [a(a-1)(a-2) q_{0,a-3}(x)] := \begin{cases}
x^{2}-1, & \mbox{if } a = 0, \\
-x, & \mbox{if } a = 1, \\
2, & \mbox{if } a = 2, \\
a(a-1)(a-2) q_{0,a-3}(x), & \mbox{if } a \geq 3.
\end{cases} \label{q0 a negative}
\end{align}
To be concrete, the first polynomials are given by 
\begin{align}
	& p_{0,a}(x) = \Big\{1, \; x, \; x^{2}+1, \; x^{3}+3x, \; x^{4}+6x^{2}+3 \Big\}, \quad  q_{0,a}(x) = \Big\{0, \; 1, \; x, \; x^{2}+2, \; x^{3}+5x \Big\}, 
	\\
	&  p_{1,a}(x) =- \frac{a}{2} p_{0,a+1}(x)  +b\, \Big\{  0, -x^2+1, -2x^3+4x, -3x^4+6x^2+5,-4x^5+4x^3+32x  \Big\},
	\\
	& q_{1,a}(x) = - \frac{a}{2} q_{0,a+1}(x) + b \, \Big\{ -\frac{5}{3}x^2+\frac23,  \frac23 x, -2x^2+\frac83,  -3x^3+9x , -4x^4+8x^2+16 \Big\}.
\end{align}
In the above, the $5$ polynomials inside the brackets correspond, from left to right, to $a=0,1,2,3,4$. 

\medskip The large $n$ asymptotics of $\mathbb{E}\big[ e^{\frac{u}{\pi} \, \im \ln p_{n}(r)}e^{a \, \re \ln p_{n}(r)} \big]$ are naturally described in terms of the two functions
\begin{align}
& \mathcal{G}_{0}(y; u,a):= p_{0,a}(-\sqrt{2}y)\bigg( (-1)^{a} + \frac{e^{u}-(-1)^{a}}{2}\mathrm{erfc}(y) \bigg) + q_{0,a}(-\sqrt{2} y) (e^{u}-(-1)^{a}) \frac{e^{-y^{2}}}{\sqrt{2\pi}}, \label{function F} \\
& \mathcal{G}_{1}(y; u,a) := p_{1,a}(-\sqrt{2}y)\bigg( (-1)^{a} + \frac{e^{u}-(-1)^{a}}{2}\mathrm{erfc}(y) \bigg) + q_{1,a}(-\sqrt{2} y) (e^{u}-(-1)^{a}) \frac{e^{-y^{2}}}{\sqrt{2\pi}}, \label{function H}
\end{align}
where $y \in \mathbb{R}$, $u \in \mathbb{R}$, $a \in \mathbb{N}$, and $\mathrm{erfc}$ is the complementary error function
\begin{align}\label{def of erfc}
\mathrm{erfc} (y) = \frac{2}{\sqrt{\pi}}\int_{y}^{\infty} e^{-x^{2}}dx.
\end{align}
In the statement of our theorem, $\mathcal{G}_{0}$ appears inside a logarithm and in a denominator. It turns out that $\mathcal{G}_{0}(y; u,a)>0$ for all $y \in \mathbb{R}$, $u \in \mathbb{R}$ and $a \in \mathbb{N}$. This fact is not obvious so we defer the proof to Section \ref{section:S2}, see Lemma \ref{lemma:mathcalG0 positive}.

\begin{theorem}\label{thm:main thm}

\noindent Let $b>0$, $\alpha > -1$, $r \in (0,b^{-\frac{1}{2b}})$, $u \in \mathbb{R}$ and $a \in \mathbb{N}$ be fixed parameters. As $n \to + \infty$,  
\begin{align}\label{asymp in main thm}
\mathbb{E}\Big[ e^{\frac{u}{\pi} \, \im \ln p_{n}(r)}e^{a \, \re \ln p_{n}(r)} \Big] = \exp \bigg( C_{1} n + C_{2} \sqrt{n} + C_{3} + \bigO\Big( n^{-\frac{1}{2b}} + (\ln n)^{\frac{5}{8}}n^{-\frac{1}{8}} \Big)\bigg),
\end{align}
where
\begin{align*}
& C_{1} = \int_{0}^{r} \Big( u+a \ln(r-y) \Big)d\mu(y) + \int_{r}^{b^{-\frac{1}{2b}}} a \ln(y-r) \, d\mu(y), \\
& C_{2} = \sqrt{2}\, b r^{b} \int_{-\infty}^{+\infty} \Big( \ln \mathcal{G}_{0}(y; u,a) - a \ln (\sqrt{2}|y|) - u \chi_{(-\infty,0)}(y) \Big)\, dy,\\
& C_{3} = - \bigg( \frac{1}{2}+\alpha \bigg)u + \frac{a(1-a)}{4(1-(br^{2b})^{\frac{1}{2b}})} + \frac{a}{4}(2+a-2b+4\alpha) \ln \Big( (br^{2b})^{-\frac{1}{2b}}-1 \Big) \\
&   +  \int_{-\infty}^{+\infty} \bigg\{ \frac{1}{\sqrt{2}} \,\frac{ \mathcal{G}_1( y ; u,a) }{   \mathcal{G}_0(y; u,a)  } + 4b y \Big( \ln(\mathcal{G}_{0}(y; u,a))- u \chi_{(-\infty,0)}(y) \Big) \\
& \hspace{1.65cm} -  \frac{a}{2} y \Big( 1+2b+8b\ln (\sqrt{2}|y|) \Big)   + \frac{(2ab-a^{2})y}{4(1+y^{2})}  \bigg\} \,dy, 
\end{align*}
$d\mu(y) = 2b^{2}y^{2b-1}dy$, and $\chi_{(-\infty,0)}(y)= 0$ for $y \geq 0$ and $\chi_{(-\infty,0)}(y)=1$ for $y<0$. 
\end{theorem}

\begin{remark}
For $a=0$, the next term in \eqref{asymp in main thm} is of order $n^{-\frac{1}{2}}$ and was obtained explicitly in \cite{Charlier 2d jumps}. For $a \geq 1$, numerical simulations suggest that the $\bigO$-term is in fact of order $n^{-\frac{1}{2}} + n^{-\frac{1}{2b}}$. This also suggests that our estimate in \eqref{asymp in main thm} for the $\bigO$-term is optimal for $b>4$ and $a \geq 1$.
\end{remark}

\subsection*{Outline of the proof of Theorem \ref{thm:main thm}.}

Let $\mathcal{E}_{n}:= \mathbb{E}\big[ e^{\frac{u}{\pi} \, \im \ln p_{n}(r)}e^{a \, \re \ln p_{n}(r)} \big]$. Our starting point is the following exact formula
\begin{align}\label{main exact formula}
\ln \mathcal{E}_{n} = \sum_{j=1}^{n} \ln \bigg(\sum_{k=0}^{a} {a \choose k} \frac{(-r)^{a-k}}{n^{\frac{k}{2b}}} \frac{\Gamma(\frac{2j+2\alpha+k}{2b})}{\Gamma(\frac{2j+2\alpha}{2b})} \bigg[ 1+((-1)^{a}e^{u}-1)\frac{\gamma(\frac{2j+2\alpha+k}{2b},nr^{2b})}{\Gamma(\frac{2j+2\alpha+k}{2b})} \bigg] \bigg),
\end{align}
where $\gamma(\tilde{a},z)$ is the incomplete gamma function
\begin{align*}
\gamma(\tilde{a},z) = \int_{0}^{z}t^{\tilde{a}-1}e^{-t}dt.
\end{align*}
Formula \eqref{main exact formula} can be derived using the facts that \eqref{def of point process} is determinantal and that $w$ is rotation-invariant, see Lemma \ref{lemma:exact formula} below. For fixed $j$ and $a \geq 1$, it is easy to see that the summand in \eqref{main exact formula} contains a term proportional to $n^{-\frac{1}{2b}}$ in its large $n$ asymptotics. This already explains why our estimate for the error term in \eqref{asymp in main thm} contains $n^{-\frac{1}{2b}}$.


\medskip To obtain precise asymptotics for $\mathcal{E}_{n}$, up to and including the term $C_{3}$ of order $1$, we must take into account \textit{each} of the $n$ terms in the sum \eqref{main exact formula} (they all contribute). As can be seen from \eqref{main exact formula}, this means that we need precise uniform asymptotics for $\gamma(\tilde{a},z)$ as both $z \to + \infty$ and $\tilde{a} \to + \infty$ at various different relative speeds. Fortunately, these asymptotics are available in the literature \cite{Temme}. Following the approach of \cite{ForresterHoleProba} (which was further developed in \cite{Charlier 2d jumps, Charlier 2d gap}), we will split the sum \eqref{main exact formula} in several parts,
\begin{align*}
\ln \mathcal{E}_{n} = S_{0} + S_{1} + S_{2} + S_{3},
\end{align*}
where $S_{\ell}$, $\ell=0,1,2,3$ are given in \eqref{def of S0}--\eqref{def of S3}. There is a critical transition in the large $\tilde{a}$ asymptotics of $\gamma(\tilde{a},z)$ when $z \to + \infty$ is such that $\lambda = \frac{z}{\tilde{a}}\approx 1$. The sum $S_{2}$ is the hardest one and precisely corresponds to this critical transition; it requires a ``local analysis" involving the $j$-terms in \eqref{main exact formula} for which $\smash{\frac{bnr^{2b}}{j}}\approx 1$. 

\medskip We found that, quite surprisingly, ``circular" root-type singularities are significantly more involved to analyze than ``circular" jump-type singularities. Let us highlight some of the reasons for that:
\begin{itemize}
\item The ``global analysis" needed for $S_{0}$, $S_{1}$ and $S_{3}$ requires some precise Riemann sum approximations for functions with singularities. For comparison, the analogue of $S_{0}$, $S_{1}$ and $S_{3}$ in \cite{Charlier 2d jumps} in the case of pure ``circular" jump-type singularities are straightforward to analyze, because the corresponding Riemann sum approximations only involve constant functions.
\item Huge cancellations occur in the ``local analysis" of $S_{2}$. In fact, to obtain $C_{3}$, we need to expand up to the $(a+2)$-th order the summand of the $k$-sum in \eqref{main exact formula}. This is because, curiously, the first $a$ terms in the expansion cancel perfectly after summing over $k$. To treat the general case $a \in \mathbb{N}$, this means that we need to expand various quantities to all orders. An important technical obstacle is that the coefficients in these various expansions are not always readily available in an explicit form; sometimes they can only be found recursively and involve heavy combinatorics, see e.g. Lemma \ref{lemma:some computation for S2kp2p} and the all-order expansion of $\gamma$ in Lemma \ref{lemma: uniform}. The analysis of $S_{2}$ is in fact the only part in the proof where solving the problem for general $a \in \mathbb{N}$ is clearly harder than solving the problem for a finite number of values of $a$, say $a \in \{0,1,2,3,4\}$. This is also the only place in the proof where the (associated) Hermite polynomials arise, see Lemmas \ref{lemma:sum of qaa(1)} and \ref{lemma:sum of qaa(4)}.
\end{itemize}

\begin{remark}
For non-integer values of $a$, formula \eqref{main exact formula} does not hold and the connection with the incomplete gamma function is lost (and therefore the strong results from \cite{Temme} cannot be used anymore). This is the main reason as to why we decided to restrict ourselves to $a \in \mathbb{N}$ in this work. For $a \notin \mathbb{N}$, the exact expression for $\mathcal{E}_{n}$ involves hypergeometric functions that generalize the incomplete gamma function. Also, because of the well-known relation $D_{k}(z) = e^{-\frac{z^{2}}{4}}\He_{k}(z)$, $k\in \mathbb{N}$ (see \cite[eq 12.7.2]{NIST}), it is tempting to conjecture that for the general case $a \in (-1,+\infty)$ the large $n$ asymptotics of $\mathcal{E}_{n}$ involve the parabolic cylinder function. That would be very interesting to figure that out in detail and we intend to come back to this problem in a future publication.
\end{remark}

\paragraph{Outline of the paper.} In Section \ref{section:proof}, we prove \eqref{main exact formula}, define the sums $S_{j}$, $j=0,1,2,3$, and establish many useful lemmas. In Section \ref{section: S0S3S1}, we obtain the large $n$ asymptotics of $S_{0}$, $S_{3}$ and $S_{1}$. The large $n$ asymptotics of $S_{2}$ are then obtained in Section \ref{section:S2}. We finish the proof of Theorem \ref{thm:main thm} in Section \ref{section:proof thm}.

\section{Preliminaries}\label{section:proof}
This section contains the proof of \eqref{main exact formula} and the definitions of $S_{0},\ldots,S_{3}$. We also establish here various preliminary lemmas that will be used in Sections \ref{section: S0S3S1} and \ref{section:S2}. 
\begin{lemma}\label{lemma:exact formula}
Formula \eqref{main exact formula} holds for all $n \in \mathbb{N}_{>0}:=\{1,2,\ldots\}$.
\end{lemma}
\begin{proof}
The partition function $Z_{n}$ of the Mittag-Leffler ensemble is known to be
\begin{align}\label{Zn as product}
Z_{n} = n^{-\frac{n^{2}}{2b}}n^{-\frac{1+2\alpha}{2b}n} \frac{\pi^{n}}{b^{n}} \prod_{j=1}^{n} \Gamma(\tfrac{j+\alpha}{b}),
\end{align}
see e.g. \cite[eq. (1.23)]{Charlier 2d jumps}.
Since $\mathcal{E}_{n}=D_{n}/Z_{n}$, it only remains to find a simplified exact expression for $D_{n}$. Since $w$ is rotation-invariant, $\int_{\mathbb{C}} z^{j} \overline{z}^{k} w(z) d^{2}z = 0$ for $j \neq k$, and therefore, by \eqref{def of Dn intro}, we have
\begin{align*}
D_{n} = \prod_{j=0}^{n-1} \int_{\mathbb{C}} |z|^{2j} w(z) d^{2}z = (2\pi)^{n} \prod_{j=1}^{n} \int_{0}^{+\infty}v^{2j-1}w(v)dv.
\end{align*}
Since $a \in \mathbb{N}$, 
\begin{align*}
w(v) = v^{2\alpha}e^{-nv^{2b}} \begin{cases}
\ds e^{u}\sum_{k=0}^{a} {a \choose k} (-1)^{k}v^{k}r^{a-k}, & \mbox{if } v<r, \\
\ds \sum_{k=0}^{a} {a \choose k} (-1)^{a-k}v^{k}r^{a-k}, & \mbox{if } v\geq r,
\end{cases}
\end{align*}
and thus
\begin{align}\label{Dn as product}
& D_{n} = n^{-\frac{n^{2}}{2b}}n^{-\frac{1+2\alpha}{2b}n} \frac{\pi^{n}}{b^{n}} \prod_{j=1}^{n} \sum_{k=0}^{a} {a \choose k} \frac{(-r)^{a-k}}{n^{\frac{k}{2b}}} \bigg(\Gamma(\tfrac{2j+2\alpha+k}{2b}) + ((-1)^{a}e^{u}-1) \gamma(\tfrac{2j+2\alpha+k}{2b},nr^{2b}) \bigg).
\end{align}
The claim now follows directly from \eqref{Zn as product}, \eqref{Dn as product} and $\mathcal{E}_{n}=D_{n}/Z_{n}$.
\end{proof}

Through the paper, $c$ and $C$ denote positive constants which may change within a computation, and $\ln$ always denotes the principal branch of the logarithm. 

\medskip Let $M'$ be a large integer independent of $n$, let $\epsilon > 0$ be a small constant independent of $n$, and let $M:= n^{\frac{1}{8}}(\ln n)^{-\frac{1}{8}}$. Define
\begin{align}\label{def of jminus and jplus}
& j_{-}:=\lceil \tfrac{bnr^{2b}}{1+\epsilon} - \alpha \rceil, \qquad j_{+} := \lfloor  \tfrac{bnr^{2b}}{1-\epsilon} - \alpha \rfloor.
\end{align}
We choose $\epsilon$ small enough so that
\begin{align*}
\frac{br^{2b}}{1-\epsilon} < \frac{1}{1+\epsilon}.
\end{align*}
Using \eqref{main exact formula}, we divide $\ln \mathcal{E}_{n}$ into $4$ parts
\begin{align}\label{log Dn as a sum of sums}
\ln \mathcal{E}_{n} = S_{0} + S_{1} + S_{2} + S_{3},
\end{align}
with 
\begin{align}
& S_{0} = \sum_{j=1}^{M'} \ln \bigg( \sum_{k=0}^{a} {a \choose k} \frac{(-r)^{a-k}}{n^{\frac{k}{2b}}} \frac{\Gamma(\frac{2j+2\alpha+k}{2b})}{\Gamma(\frac{2j+2\alpha}{2b})} \bigg[ 1+((-1)^{a}e^{u}-1)\frac{\gamma(\frac{2j+2\alpha+k}{2b},nr^{2b})}{\Gamma(\frac{2j+2\alpha+k}{2b})} \bigg] \bigg), \label{def of S0} \\
& S_{1} = \sum_{j=M'+1}^{j_{-}-1} \hspace{-0.25cm} \ln \bigg( \sum_{k=0}^{a} {a \choose k} \frac{(-r)^{a-k}}{n^{\frac{k}{2b}}} \frac{\Gamma(\frac{2j+2\alpha+k}{2b})}{\Gamma(\frac{2j+2\alpha}{2b})} \bigg[ 1+((-1)^{a}e^{u}-1)\frac{\gamma(\frac{2j+2\alpha+k}{2b},nr^{2b})}{\Gamma(\frac{2j+2\alpha+k}{2b})} \bigg] \bigg), \label{def of S1} \\
& S_{2} = \sum_{j=j_{-}}^{j_{+}} \ln \bigg( \sum_{k=0}^{a} {a \choose k} \frac{(-r)^{a-k}}{n^{\frac{k}{2b}}} \frac{\Gamma(\frac{2j+2\alpha+k}{2b})}{\Gamma(\frac{2j+2\alpha}{2b})} \bigg[ 1+((-1)^{a}e^{u}-1)\frac{\gamma(\frac{2j+2\alpha+k}{2b},nr^{2b})}{\Gamma(\frac{2j+2\alpha+k}{2b})} \bigg] \bigg), \label{def of S2} \\
& S_{3}=\sum_{j=j_{+}+1}^{n} \hspace{-0.2cm} \ln \bigg( \sum_{k=0}^{a} {a \choose k} \frac{(-r)^{a-k}}{n^{\frac{k}{2b}}} \frac{\Gamma(\frac{2j+2\alpha+k}{2b})}{\Gamma(\frac{2j+2\alpha}{2b})} \bigg[ 1+((-1)^{a}e^{u}-1)\frac{\gamma(\frac{2j+2\alpha+k}{2b},nr^{2b})}{\Gamma(\frac{2j+2\alpha+k}{2b})} \bigg] \bigg). \label{def of S3}
\end{align}
In Sections \ref{section: S0S3S1} and \ref{section:S2}, we will analyze these sums in order of increasing difficulty: first $S_{0}$, then $S_{3}$, then $S_{1}$, and finally $S_{2}$. The sum $S_{0}$ is straightforward to analyze, but $S_{1}$, $S_{2}$ and $S_{3}$ are more involved and require some preparation. This preparation is carried out in the next subsection.


\subsection{Useful lemmas}
For $\ell \in \mathbb{N}:=\{0,1,\ldots\}$, let
\begin{equation}\label{def of gfrakl}
\mathfrak{g}_\ell(x):= \sum_{k=0}^a \binom{a}{k} x^k k^\ell. 
\end{equation}
If $k=\ell=0$ in \eqref{def of gfrakl}, then $k^{\ell}:=1$ so that $\mathfrak{g}_{0}(x)=(1+x)^{a}$. 
\begin{remark}
The sequence $\{\mathfrak{g}_{\ell}\}_{\ell=0}^{+\infty}$ satisfies $\mathfrak{g}_{\ell+1}(x)= x \mathfrak{g}_\ell'(x)$, $\ell\in \mathbb{N}$. Solving this recurrence relation using the initial value $\mathfrak{g}_{0}(x)=(1+x)^{a}$ yields
\begin{equation*}
\mathfrak{g}_\ell(x)= \Big[ \frac{d}{dt}\Big]^{\ell} \Big[ (1+e^t)^a  \Big] \Big|_{t=\ln x},
\end{equation*}
which is an interesting alternative representation of $\mathfrak{g}_{\ell}$. 
\end{remark}
The next lemma establishes yet another representation of $\mathfrak{g}_{\ell}$.
\begin{lemma}\label{lemma:prop of gfrakl}
For $\ell \in \mathbb{N}_{>0}$, we have
\begin{align}\label{gfrakl in lemma}
\mathfrak{g}_{\ell}(x) = x (x+1)^{a-\min(\ell,a)} \sum_{j=1}^{\min(\ell,a)} S(\ell,j) \frac{a!}{(a-j)!} x^{j-1}(x+1)^{\min(\ell,a)-j},
\end{align}
where $S(\ell,j)$ is the Stirling number of the second kind, i.e. the number of partitions of $\{1, \dots, \ell\}$ into exactly $j$ nonempty subsets. Furthermore,
\begin{align}
& \mathfrak{g}_{\ell}(x) = \bigO(x), & & \mbox{as } x \to 0, \quad \ell \in \mathbb{N}_{>0}, \label{estimate for gfrak ell 0} \\
& \mathfrak{g}_{\ell}(x) = \bigO(\min\{1,|x+1|^{a-\ell}\}), & & \mbox{as } x \to -1, \quad \ell\in \mathbb{N}. \label{estimate for gfrak ell}
\end{align}
\end{lemma}
\begin{remark}
Since $\mathfrak{g}_{0}(x)=(1+x)^{a}$, \eqref{gfrakl in lemma} and \eqref{estimate for gfrak ell 0} do not hold for $\ell=0$.
\end{remark}
\begin{proof}
Let $\ell \in \mathbb{N}_{>0}$ be fixed. By \cite[eq. 26.8.10]{NIST}, we have
\begin{equation}\label{def of alj}
k^{\ell}= \sum_{j=1}^\ell (k)_{j} S(\ell,j), \qquad \mbox{for all } k \in \mathbb{C},
\end{equation}
where $(k)_{j}:= k(k-1)(k-2)\ldots (k-j+1)$ is the descending factorial. Substituting \eqref{def of alj} in \eqref{def of gfrakl}, we obtain
\begin{align}
\mathfrak{g}_\ell(x) & = \sum_{k=0}^a \sum_{j=1}^{\min(\ell,k)} \frac{a! S(\ell,j) x^k}{(a-k)!(k-j)!} = \sum_{j=1}^{\min(\ell,a)}\sum_{k=j}^a  \frac{a! S(\ell,j) x^k}{(a-k)!(k-j)!} \nonumber \\
& = \sum_{j=1}^{\min(\ell,a)} S(\ell,j) \frac{a!}{(a-j)!} \sum_{k=j}^a  \binom{a-j}{k-j}x^k = \sum_{j=1}^{\min(\ell,a)} S(\ell,j) \frac{a!}{(a-j)!} x^{j}(x+1)^{a-j}, \label{gfrakl in the proof}
\end{align}
which is \eqref{gfrakl in lemma}. The expansions \eqref{estimate for gfrak ell 0} and \eqref{estimate for gfrak ell} for $\ell \geq 1$ directly follows from \eqref{gfrakl in lemma}, and \eqref{estimate for gfrak ell} for $\ell=0$ follows from $\mathfrak{g}_{0}(x)=(1+x)^{a}$. 
\end{proof}
The sums $S_{1}$, $S_{2}$ and $S_{3}$ naturally involve the functions
\begin{align} \label{def of gammaell}
	\gamma_{\ell}(x) := \begin{cases}
\ds \sum_{k=0}^{a} {a \choose k} (-r)^{a-k} \bigg( \frac{x}{b} \bigg)^{\frac{k}{2b}}k^{\ell}, & x > b r^{2b}, \\
\ds \sum_{k=0}^{a} {a \choose k} (-1)^{ k } r^{a-k} \bigg( \frac{x}{b} \bigg)^{\frac{k}{2b}}k^{\ell}, & x \in (0,b r^{2b}).
\end{cases}
\end{align}
The next lemma collects some properties of $\gamma_{\ell}$.
\begin{lemma}\label{lemma:gammal}
Let $\ell \in \mathbb{N}$. The function $\gamma_{\ell}$ can be written as
\begin{align}\label{def of gammaj Stirling}
\gamma_{\ell}(x) = \begin{cases}
\ds \Big|r-\Big(\frac{x}{b}\Big)^{\frac{1}{2b}}\Big|^{a}, & \mbox{if } \ell=0, \\
\ds \Big(\frac{x}{b}\Big)^{\frac{1}{2b}} \Big|r-\Big(\frac{x}{b}\Big)^{\frac{1}{2b}}\Big|^{a} \sum_{j=1}^{\min(\ell,a)} \frac{a!S(\ell,j)}{(a-j)!} \Big(\frac{x}{b}\Big)^{\frac{j-1}{2b}} \Big( \Big(\frac{x}{b}\Big)^{\frac{1}{2b}}-r  \Big)^{-j}, & \mbox{if } \ell \geq 1.
\end{cases}
\end{align}
In particular,
\begin{align}
& \gamma_{\ell}(x) = \bigO(x^{\frac{1}{2b}}), & & \mbox{as } x \to 0, \quad \ell \in \mathbb{N}_{>0}, \label{estimate for gamma ell 0} \\
& \gamma_{\ell}(x) = \bigO(\min\{1,|x-br^{2b}|^{a-\ell}\}), & & \mbox{as } x \to br^{2b}, \quad \ell\in \mathbb{N}. \label{estimate for gamma ell}
\end{align}
\end{lemma}
\begin{proof}
It is easily checked that
\begin{align} \label{def of gammaell in proof}
\gamma_{\ell}(x) = \begin{cases}
\displaystyle (-r)^a \mathfrak{g}_\ell\Big( -\Big( \frac{x}{br^{2b}}\Big)^{\frac{1}{2b}} \Big) 	, & x > b r^{2b}, \smallskip \\
r^a \mathfrak{g}_\ell\Big( -\Big( \dfrac{x}{br^{2b}}\Big)^{\frac{1}{2b}} \Big) 	, & x \in (0,b r^{2b}).
\end{cases}
\end{align}
The claim is now a straightforward consequence of Lemma \ref{lemma:prop of gfrakl}. 
\end{proof}

\begin{lemma}\label{lemma:ratio of gamma}
Let $k \in \mathbb{N}$ be fixed. As $j \to + \infty$, 
\begin{align}\label{large j asymp for the ratio of Gamma}
\frac{\Gamma(\frac{2j+2\alpha+k}{2b})}{\Gamma(\frac{2j+2\alpha}{2b})} \sim \bigg( \frac{j}{b} \bigg)^{\frac{k}{2b}} \bigg( 1+ \sum_{\ell=1}^{+\infty} \frac{\mathfrak{p}_{2\ell}(k)}{j^{\ell}}\bigg),
\end{align}
where
\begin{align}\label{def of p2l}
\mathfrak{p}_{2\ell}(k) := b^\ell \binom{\frac{k}{2b}}{\ell} B_\ell^{(1+ \frac{k}{2b}) }\Big( \frac{2\alpha+k}{2b}\Big) =: \sum_{m=1}^{2\ell} \mathfrak{p}_{2\ell,m}\,k^m,
\end{align}
$\binom{\frac{k}{2b}}{\ell} := \frac{\frac{k}{2b}(\frac{k}{2b}-1)\ldots (\frac{k}{2b}-\ell+1)}{\ell!}$, and $B_{\ell}^{(k)}(x)$ is the generalized Bernoulli\footnote{$B_\ell^{(1)}(x)$ is the classical Bernoulli polynomial of degree $\ell$, and $B_\ell^{(0)}(x) = x^{\ell}$.}  polynomial of degree $\ell$ defined through the generating function
\begin{align}\label{def of gen Bernouilli}
\Big( \frac{t}{e^t-1}\Big)^k \,e^{xt}= \sum_{\ell=0}^{+\infty} B_\ell^{(k)}(x) \frac{t^n}{n!}, \qquad |t|<2\pi.
\end{align}
\end{lemma}
\begin{remark}
The degree $2\ell$ polynomial $\mathfrak{p}_{2\ell}$ satisfies $\mathfrak{p}_{2\ell}(0)=0$. This is consistent with the fact that for $k=0$ the left-hand side of \eqref{large j asymp for the ratio of Gamma} is $1$.
\end{remark}
\begin{proof}
The claim directly follow from \cite[eq. 5.11.13]{NIST}
\begin{align}\label{lol33}
\frac{\Gamma(v+p_{2}+p_{1})}{\Gamma(v+p_{2})} \sim v^{p_{1}} \sum_{\ell=0}^{+\infty} \binom{p_{1}}{\ell} \frac{B_{\ell}^{(p_{1}+1)}(p_{1}+p_{2})}{v^{\ell}} \qquad \mbox{as } v \to + \infty, \; p_{1},p_{2} \mbox{ fixed}.
\end{align}
\end{proof}
\begin{lemma}
As $n \to + \infty$, 
\begin{align}
& \sum_{k=0}^{a} {a \choose k} \frac{(-r)^{a-k}}{n^{\frac{k}{2b}}} \frac{\Gamma(\frac{2j+2\alpha+k}{2b})}{\Gamma(\frac{2j+2\alpha}{2b})} \sim \gamma_{0}(j/n) + \sum_{\ell=1}^{+\infty} \frac{1}{n^{\ell}} \frac{\sum_{m=1}^{2\ell}\mathfrak{p}_{2\ell,m}\gamma_{m}(j/n)}{(j/n)^{\ell}}, \label{lol1} 
\end{align}
uniformly for $j \in \{j_{+}+1,\ldots,n\}$. Furthermore, $M'$ can be chosen sufficiently large (but fixed) such that the following holds: there exists $C>0$ such that 
\begin{align}
& \bigg|\sum_{k=0}^{a}  {a \choose k} \hspace{-0.07cm} \frac{r^{a-k}(-1)^{k}}{n^{\frac{k}{2b}}} \frac{\Gamma(\frac{2j+2\alpha+k}{2b})}{\Gamma(\frac{2j+2\alpha}{2b})} \hspace{-0.04cm} - \hspace{-0.04cm} \bigg\{ \hspace{-0.05cm} \gamma_{0}(j/n) \hspace{-0.03cm} + \hspace{-0.03cm} \frac{\gamma_{2}(j/n)\hspace{-0.03cm}+\hspace{-0.03cm}(4\alpha-2b)\gamma_{1}(j/n)}{8bj} \bigg\} \bigg|   \leq C \frac{(j/n)^{\frac{1}{2b}-2}}{n^{2}}, \label{lol5}
\end{align}
for all sufficiently large $n$ and all $j \in \{M'+1,\ldots,j_{-}-1\}$.
\end{lemma}
\begin{proof}
Since $a$ is fixed, Lemma \ref{lemma:ratio of gamma} implies that
\begin{align}\label{lol38}
& \sum_{k=0}^{a} {a \choose k} \frac{(-r)^{a-k}}{n^{\frac{k}{2b}}} \frac{\Gamma(\frac{2j+2\alpha+k}{2b})}{\Gamma(\frac{2j+2\alpha}{2b})} \sim \sum_{k=0}^{a} {a \choose k} (-r)^{a-k} \bigg( \frac{j/n}{b} \bigg)^{\frac{k}{2b}} \bigg( 1+ \sum_{\ell=1}^{+\infty} \frac{\mathfrak{p}_{2\ell}(k)}{j^{\ell}}\bigg),
\end{align}
as $j \to + \infty$. Since $j/n \in (br^{2b},1]$ for all $j \in \{j_{+}+1,\ldots,n\}$, the expansion \eqref{lol1} directly follows from \eqref{lol38} and the definition \eqref{def of gammaell} of $\gamma_{\ell}$. In the same way, but using now the definition \eqref{def of gammaell} of $\gamma_{\ell}(x)$ for $x \in (0,br^{2b})$, we infer that
\begin{align}\label{lol37}
& \sum_{k=0}^{a} {a \choose k} \frac{r^{a-k}(-1)^{k}}{n^{\frac{k}{2b}}} \frac{\Gamma(\frac{2j+2\alpha+k}{2b})}{\Gamma(\frac{2j+2\alpha}{2b})} \sim \gamma_{0}(j/n) + \sum_{\ell=1}^{+\infty} \frac{\sum_{m=1}^{2\ell}\mathfrak{p}_{2\ell,m}\gamma_{m}(j/n)}{j^{\ell}},
\end{align}
as $j \to + \infty$ uniformly for $n$ such that $j/n \in (0,br^{2b})$. The estimate \eqref{lol5} then follows from \eqref{lol37} and the behavior \eqref{estimate for gamma ell 0}.
\end{proof}

For the large $n$ analysis of $S_{3},S_{1},S_{2}$, we will need to approximate various large sums involving functions with singularities; for this we will also rely on the following lemma from \cite{Charlier 2d gap}.
\begin{lemma}\label{lemma:Riemann sum NEW}(\cite[Lemma 3.4]{Charlier 2d gap})
Let $A,a_{0}$, $B,b_{0}$ be bounded function of $n \in \{1,2,\ldots\}$, such that 
\begin{align*}
& a_{n} := An + a_{0} \qquad \mbox{ and } \qquad b_{n} := Bn + b_{0}
\end{align*}
are integers. Assume also that $B-A$ is positive and remains bounded away from $0$. Let $f$ be a function independent of $n$, and which is $C^{4}([\min\{\frac{a_{n}}{n},A\},\max\{\frac{b_{n}}{n},B\}])$ for all $n\in \{1,2,\ldots\}$. Then as $n \to + \infty$, we have
\begin{align}
&  \sum_{j=a_{n}}^{b_{n}}f(\tfrac{j}{n}) = n \int_{A}^{B}f(x)dx + \frac{(1-2a_{0})f(A)+(1+2b_{0})f(B)}{2}  \nonumber \\
& + \frac{(-1+6a_{0}-6a_{0}^{2})f'(A)+(1+6b_{0}+6b_{0}^{2})f'(B)}{12n}+ \frac{(-a_{0}+3a_{0}^{2}-2a_{0}^{3})f''(A)+(b_{0}+3b_{0}^{2}+2b_{0}^{3})f''(B)}{12n^{2}} \nonumber \\
& + \bigO \bigg( \frac{\mathfrak{m}_{A}(f''')+\mathfrak{m}_{B}(f''')}{n^{3}} + \sum_{j=a_{n}}^{b_{n}-1} \frac{\mathfrak{m}_{j,n}(f'''')}{n^{4}} \bigg), \label{sum f asymp gap NEW}
\end{align}
where, for a given function $g$ continuous on $[\min\{\frac{a_{n}}{n},A\},\max\{\frac{b_{n}}{n},B\}]$,
\begin{align*}
\mathfrak{m}_{A}(g) := \max_{x \in [\min\{\frac{a_{n}}{n},A\},\max\{\frac{a_{n}}{n},A\}]}|g(x)|, \qquad \mathfrak{m}_{B}(g) := \max_{x \in [\min\{\frac{b_{n}}{n},B\},\max\{\frac{b_{n}}{n},B\}]}|g(x)|,
\end{align*}
and for $j \in \{a_{n},\ldots,b_{n}-1\}$, $\mathfrak{m}_{j,n}(g) := \max_{x \in [\frac{j}{n},\frac{j+1}{n}]}|g(x)|$.
\end{lemma}
\section{Global analysis: large $n$ asymptotics of $S_{0}$, $S_{3}$ and $S_{1}$}\label{section: S0S3S1}
As mentioned earlier, we will analyze the sums \eqref{def of S0}--\eqref{def of S3} in order of increasing difficulty: first $S_{0}$, then $S_{3}$, then $S_{1}$, and finally $S_{2}$. In this section we focus on $S_{0}$, $S_{3}$ and $S_{2}$. We defer the analysis of $S_{2}$ to the next section.
\begin{lemma}\label{lemma: S0}
As $n \to + \infty$, 
\begin{align}\label{asymp of S0}
S_{0} = M' \ln (r^{a}e^{u}) + \bigO(n^{-\frac{1}{2b}} ).
\end{align}
\end{lemma}
\begin{proof}
Using \eqref{def of S0} and Lemma \ref{lemma:various regime of gamma}, we obtain 
\begin{align*}
S_{0} & = \sum_{j=1}^{M'} \ln \bigg( \sum_{k=0}^{a} {a \choose k} \frac{(-r)^{a-k}}{n^{\frac{k}{2b}}} \frac{\Gamma(\frac{2j+2\alpha+k}{2b})}{\Gamma(\frac{2j+2\alpha}{2b})} \bigg[ 1+((-1)^{a}e^{u}-1)[1 + \bigO(e^{-cn})] \bigg] \bigg), \\
& = \sum_{j=1}^{M'} \ln \bigg( \sum_{k=0}^{a} {a \choose k} \frac{(-1)^{k}r^{a-k}e^{u}}{n^{\frac{k}{2b}}} \frac{\Gamma(\frac{2j+2\alpha+k}{2b})}{\Gamma(\frac{2j+2\alpha}{2b})} \bigg) + \bigO(e^{-cn}), \qquad \mbox{as } n \to + \infty.
\end{align*}
Since $M'$ is fixed, only the $(k=0)$-terms contribute to order $1$ in the large $n$ asymptotics of $S_{0}$; the other terms are $\bigO(n^{-\frac{1}{2b}})$. 
\end{proof}
Recall that $S_{1},S_{3}$ are given by \eqref{def of S1} and \eqref{def of S3}. Following the approach of \cite{Charlier 2d jumps, Charlier 2d gap}, we define
\begin{align} \label{def of theta n eps}
\theta_{+}^{(n,\epsilon)} = \bigg( \frac{b n r^{2b}}{1-\epsilon}-\alpha \bigg)-\bigg\lfloor \frac{b n r^{2b}}{1-\epsilon}-\alpha \bigg\rfloor, \qquad \theta_{-}^{(n,\epsilon)} = \bigg\lceil \frac{b n r^{2b}}{1+\epsilon}-\alpha \bigg\rceil-\bigg( \frac{b n r^{2b}}{1+\epsilon}-\alpha \bigg),
\end{align}
and for $j=1,\ldots,n$ and $k=0,1,\ldots,a$, we also define
\begin{align}
& a_{j}:=\frac{j+\alpha}{b}, & & \lambda_{j} := \frac{bnr^{2b}}{j+\alpha}, & & \eta_{j} := (\lambda_{j}-1)\sqrt{\frac{2 (\lambda_{j}-1-\ln \lambda_{j})}{(\lambda_{j}-1)^{2}}}, \label{def of aj} \\
& a_{j,k} := \frac{2j+2\alpha+k}{2b}, & & \lambda_{j,k} := \frac{bnr^{2b}}{j+\alpha+\frac{k}{2}}, & & \eta_{j,k} := (\lambda_{j,k}-1)\sqrt{\frac{2 (\lambda_{j,k}-1-\ln \lambda_{j,k})}{(\lambda_{j,k}-1)^{2}}}. \label{def of ajk}
\end{align}
\begin{lemma}\label{lemma: S3}
As $n \to + \infty$, 
\begin{align*}
 S_{3} & = n  \int_{\frac{br^{2b}}{1-\epsilon}}^{1} a \ln \bigg( \bigg( \frac{x}{b} \bigg)^{\frac{1}{2b}}-r \bigg)dx + a \frac{(2\alpha + 2\theta_{+}^{(n,\epsilon)}-1)\ln(r(1-\epsilon)^{-\frac{1}{2b}}-r)+\ln(b^{-\frac{1}{2b}}-r)}{2} \\
& + \frac{a}{4}\bigg( \frac{1-a}{(br^{2b})^{-\frac{1}{2b}}-1}-\frac{1-a}{(1-\epsilon)^{-\frac{1}{2b}}-1}+(a-2b+4\alpha)\ln \bigg( \frac{(br^{2b})^{-\frac{1}{2b}}-1}{(1-\epsilon)^{-\frac{1}{2b}}-1} \bigg) \bigg) + \bigO(n^{-1}).
\end{align*}
\end{lemma}

\begin{proof}
Using \eqref{def of S3} and Lemma \ref{lemma: uniform} (ii) with $a$ and $\lambda$ replaced by $a_{j,k}$ and $\lambda_{j,k}$ respectively, where $j \in \{j_{+}+1,\ldots,n\}$ and $k\in \{0,\ldots,a\}$, we obtain
\begin{align}\label{lol9}
S_{3} & = \sum_{j=j_{+}+1}^{n} \ln \bigg( \sum_{k=0}^{a} {a \choose k} \frac{(-r)^{a-k}}{n^{\frac{k}{2b}}} \frac{\Gamma(\frac{2j+2\alpha+k}{2b})}{\Gamma(\frac{2j+2\alpha}{2b})} \bigg[ 1+((-1)^{a}e^{u}-1)\bigO(e^{-\frac{a_{j,k}\eta_{j,k}^{2}}{2}}) \bigg] \bigg),
\end{align}
as $n \to + \infty$. It is easy to check from \eqref{def of ajk} that $c_{1} n \leq  a_{j,k} \leq c_{1}'n$, $c_{2} \leq |\lambda_{j,k}-1| \leq c_{2}'$ and $c_{3} \leq \eta_{j,k}^{2} \leq c_{3}'$ hold for some positive constants $\{c_{j},c_{j}'\}_{j=1}^{3}$, for all $n$ sufficiently large, for all $j \in \{j_{+}+1,\ldots,n\}$ and for all $k \in \{0,\ldots,a\}$. Thus
\begin{align}\label{lol2}
S_{3} & = \sum_{j=j_{+}+1}^{n} \ln \bigg( \sum_{k=0}^{a} {a \choose k} \frac{(-r)^{a-k}}{n^{\frac{k}{2b}}} \frac{\Gamma(\frac{2j+2\alpha+k}{2b})}{\Gamma(\frac{2j+2\alpha}{2b})} \bigg) + \bigO(e^{-cn}), \qquad \mbox{as } n \to + \infty.
\end{align}
To complete the proof of this lemma, we need the following weaker version of \eqref{lol1}:
\begin{align*}
\sum_{k=0}^{a} {a \choose k} \frac{(-r)^{a-k}}{n^{\frac{k}{2b}}} \frac{\Gamma(\frac{2j+2\alpha+k}{2b})}{\Gamma(\frac{2j+2\alpha}{2b})} = \gamma_{0}(j/n) + \frac{1}{n}\frac{\gamma_{2}(j/n)+(4\alpha-2b)\gamma_{1}(j/n)}{8bj/n} + \bigO(n^{-2}),
\end{align*}
as $n \to +\infty$ and simultaneously $j \in \{j_{+}+1,\ldots,n\}$. Note from \eqref{def of jminus and jplus} that $j/n$ lies in $(br^{2b},1]$ and remains bounded away from $br^{2b}$ as $n \to +\infty$ and simultaneously $j \in \{j_{+}+1,\ldots,n\}$; in particular $\gamma_{0}(j/n)$ remains bounded away from $0$. Hence, by substituting the above expansion in \eqref{lol2} and using \eqref{def of gammaj Stirling} with $\ell=0,1,2$, we obtain after a computation that
\begin{align}
& S_{3} = \Sigma_{0}+\frac{1}{n}\Sigma_{1}+\bigO(n^{-1}), \qquad \Sigma_{\ell} := \sum_{j=j_{+}+1}^{n}f_{\ell}(j/n), \; \ell=0,1, \label{def of f0 and f1}
\\
& f_{0}(x) :=\ln \gamma_0(x)= a \ln \bigg( \bigg(\frac{x}{b}\bigg)^{\frac{1}{2b}}-r \bigg), \nonumber \\
& f_{1}(x) := \frac{\gamma_2(x)+(4\alpha-2b) \gamma_1(x)}{8bx\,\gamma_0(x)} = \frac{a(\frac{x}{b})^{\frac{1}{2b}}}{8bx((\frac{x}{b})^{\frac{1}{2b}}-r)^{2}}\bigg( (a+4\alpha-2b) \bigg(\frac{x}{b}\bigg)^{\frac{1}{2b}} - (1+4\alpha-2b)r \bigg), \nonumber
\end{align}
where we have also used that $\ln(1+x)=x+\bigO(x^{2})$ as $x \to 0$. From Lemma \ref{lemma:Riemann sum NEW} (with $A=\frac{br^{2b}}{1-\epsilon}$, $a_{0} = 1-\alpha-\theta_{+}^{(n,\epsilon)}$, $B=1$ and $b_{0}=0$), we infer that
\begin{align*}
\Sigma_{\ell} = n \int_{\frac{br^{2b}}{1-\epsilon}}^{1}f_{\ell}(x)dx + \frac{(2\alpha+2\theta_{+}^{(n,\epsilon)}-1)f_{\ell}(\frac{br^{2b}}{1-\epsilon})+f_{\ell}(1)}{2} + \bigO(n^{-1}), \qquad \ell=0,1,
\end{align*}
as $n \to + \infty$. We then obtain the claim for $S_{3}$ after a computation using the simplification
\begin{align*}
\int_{\frac{br^{2b}}{1-\epsilon}}^{1}f_{1}(x)dx = \frac{a}{4}\bigg( \frac{1-a}{(br^{2b})^{-\frac{1}{2b}}-1}-\frac{1-a}{(1-\epsilon)^{-\frac{1}{2b}}-1}+(a-2b+4\alpha)\ln \bigg( \frac{(br^{2b})^{-\frac{1}{2b}}-1}{(1-\epsilon)^{-\frac{1}{2b}}-1} \bigg) \bigg).
\end{align*}
\end{proof}

The large $n$ asymptotics of $S_{1}$ are harder to obtain than those of $S_{3}$. The main reason for it is that $S_{1}$ involves small $j$'s, and for such $j$'s the quantities $\gamma_{\ell}(j/n)$ have a singular behavior, see \eqref{estimate for gamma ell 0}. This is also the reason why the error term in Lemma \ref{lemma: S1} is more complicated than in Lemma \ref{lemma: S3}.

\begin{lemma}\label{lemma: S1}
As $n \to + \infty$, 
\begin{align*}
 S_{1} &= n  \int_{0}^{\frac{br^{2b}}{1+\epsilon}} \bigg(u + a \ln \bigg( r-\bigg( \frac{x}{b} \bigg)^{\frac{1}{2b}} \bigg)dx - M' \ln (r^{a}e^{u}) \\
&  + \frac{-(u+a\ln r)+(2\theta_{-}^{(n,\epsilon)} -1 - 2\alpha)\big(u + a \ln(r-r(1+\epsilon)^{-\frac{1}{2b}})}{2} \\
& + \frac{a}{4}\bigg( \frac{a-1}{(1+\epsilon)^{\frac{1}{2b}}-1}+(a-2b+4\alpha)\ln \big( 1-(1+\epsilon)^{-\frac{1}{2b}} \big) \bigg) + \bigO\big( n^{-\frac{1}{2b}}+n^{-1} \big).
\end{align*}
\end{lemma}

\begin{proof}
Lemma \ref{lemma: uniform} (i) implies that for any $\epsilon'>0$ there exist $A=A(\epsilon'),C=C(\epsilon')>0$ such that $|\frac{\gamma(a,z)}{\Gamma(a)}-1| \leq Ce^{-\frac{a\eta^{2}}{2}}$ for all $a \geq A$, for all $\lambda=\frac{z}{a} \in [1+\epsilon',+\infty]$, and where $\eta$ is given by \eqref{lol8}. Let us take $\epsilon'=\frac{\epsilon}{2}$ and choose $M'$ so large that $a_{j} = \frac{j+\alpha}{b} \geq A(\frac{\epsilon}{2})$ for all $j \in \{M'+1,\ldots,j_{1,-}-1\}$. Thus
\begin{align}\label{lol35}
S_{1} & = \sum_{j=M'+1}^{j_{-}-1} \ln \bigg( \sum_{k=0}^{a} {a \choose k} \frac{(-r)^{a-k}}{n^{\frac{k}{2b}}} \frac{\Gamma(\frac{2j+2\alpha+k}{2b})}{\Gamma(\frac{2j+2\alpha}{2b})} \bigg[ 1+((-1)^{a}e^{u}-1)(1+\bigO(e^{-\frac{a_{j,k}\eta_{j,k}^{2}}{2}})) \bigg] \bigg),
\end{align}
as $n \to + \infty$. From a direct analysis of \eqref{def of ajk}, we infer that $a_{j,k}\eta_{j,k}^{2}$ decreases as $j$ increases from $M'+1$ to $j_{-}-1$, and $a_{j,k}\eta_{j,k}^{2}$ decreases also as $k$ increases from $0$ to $a$. Therefore
\begin{align}\label{lol34}
\frac{a_{j,k}\eta_{j,k}^{2}}{2} \geq \frac{a_{j_{-}-1}\eta_{j_{-}-1}^{2}}{2} \geq cn, \qquad \mbox{for all } j \in \{M'+1,\ldots,j_{-}-1\}, 
\end{align}
for a small enough $c>0$. Thus
\begin{align}\label{lol4}
S_{1} & = \sum_{j=M'+1}^{j_{-}-1} \ln \bigg( \sum_{k=0}^{a} {a \choose k} \frac{r^{a-k}(-1)^{k}e^{u}}{n^{\frac{k}{2b}}} \frac{\Gamma(\frac{2j+2\alpha+k}{2b})}{\Gamma(\frac{2j+2\alpha}{2b})} \bigg) + \bigO(e^{-cn}), \qquad \mbox{as } n \to + \infty.
\end{align}
Substituting \eqref{lol5} in \eqref{lol4} and using \eqref{def of gammaj Stirling} with $\ell=0,1,2$, we obtain
\begin{align}
& S_{1} = \widetilde{\Sigma}_{0}+\frac{1}{n}\widetilde{\Sigma}_{1}+\bigO\bigg(n^{-2}\sum_{j=M'+1}^{j_{-}-1}(j/n)^{\frac{1}{2b}-2}\bigg) = \widetilde{\Sigma}_{0}+\frac{1}{n}\widetilde{\Sigma}_{1} + \bigO\big( n^{-\frac{1}{2b}}+n^{-1} \big), \nonumber \\
& \widetilde{\Sigma}_{0} := \sum_{j=M'+1}^{j_{-}-1} \widetilde{f}_{0}(j/n), \quad \widetilde{\Sigma}_{1} := \sum_{j=M'+1}^{j_{-}-1} f_{1}(j/n), \qquad \widetilde{f}_{0}(x) := u + a \ln \bigg( r-\bigg(\frac{x}{b}\bigg)^{\frac{1}{2b}} \bigg), \label{lol22}
\end{align}
as $n \to + \infty$, where $f_{1}$ is given by \eqref{def of f0 and f1}. Note that $f_{1}(x) \sim c x^{\frac{1}{2b}-1}$ as $x \to 0$; thus $f_{1}$ blows up at $0$ if $b > \frac{1}{2}$. Using now Lemma \ref{lemma:Riemann sum NEW} (with $A=\frac{M'}{n}$, $a_{0} = 1$, $B=\frac{br^{2b}}{1+\epsilon}$ and $b_{0}=\theta_{-}^{(n,\epsilon)}-1-\alpha$), we get
\begin{align*}
& \widetilde{\Sigma}_{0} = n \int_{\frac{M'}{n}}^{\frac{br^{2b}}{1+\epsilon}}\widetilde{f}_{0}(x)dx + \frac{-\widetilde{f}_{0}(\frac{M'}{n})+(2\theta_{-}^{(n,\epsilon)}-1-2\alpha)\widetilde{f}_{0}(\frac{br^{2b}}{1+\epsilon})}{2} + \bigO\big( n^{-\frac{1}{2b}}+n^{-1} \big), \\
& \frac{1}{n}\widetilde{\Sigma}_{1} =  \int_{\frac{M'}{n}}^{\frac{br^{2b}}{1+\epsilon}}f_{1}(x)dx + \bigO\big( n^{-\frac{1}{2b}}+n^{-1} \big).
\end{align*}
as $n \to + \infty$. Furthermore, by a direct analysis of $\widetilde{f}_{0}$ and $f_{1}$, 
\begin{align*}
& \widetilde{f}_{0}\bigg( \frac{M'}{n} \bigg) = u + a \ln r + \bigO\big( n^{-\frac{1}{2b}} \big), \\
&  \int_{\frac{M'}{n}}^{\frac{br^{2b}}{1+\epsilon}}f_{1}(x)\,dx = \frac{a}{4}\bigg( \frac{a-1}{(1+\epsilon)^{\frac{1}{2b}}-1} + (a-2b+4\alpha) \ln \Big( 1-(1+\epsilon)^{-\frac{1}{2b}} \Big) \bigg) + \bigO\big( n^{-\frac{1}{2b}} \big),  \\
& n \int_{\frac{M'}{n}}^{\frac{br^{2b}}{1+\epsilon}}\widetilde{f}_{0}(x) \, dx = n \int_{0}^{\frac{br^{2b}}{1+\epsilon}}\widetilde{f}_{0}(x)\,dx - (u+a \ln r)M' + \bigO\big( n^{-\frac{1}{2b}} \big),
\end{align*}
as $n \to + \infty$. The claim follows after substituting the above expansions in \eqref{lol22}. 
\end{proof}
\section{Large $n$ asymptotics of $S_{2}$}\label{section:S2}
It remains to obtain the large $n$ asymptotics of $S_{2}$, which was defined in \eqref{def of S2}. For this, let us split $S_{2}$ in three pieces,
\begin{align*}
& S_{2}=S_{2}^{(1)}+S_{2}^{(2)}+S_{2}^{(3)},
\end{align*}
where 
\begin{equation}
\label{asymp prelim of S2kpvp}
S_{2}^{(v)} := \sum_{j:\lambda_{j}\in I_{v}}  \ln \bigg( \sum_{k=0}^{a} {a \choose k} \frac{(-r)^{a-k}}{n^{\frac{k}{2b}}} \frac{\Gamma(\frac{2j+2\alpha+k}{2b})}{\Gamma(\frac{2j+2\alpha}{2b})} \bigg[ 1+((-1)^{a}e^{u}-1)\frac{\gamma(\frac{2j+2\alpha+k}{2b},nr^{2b})}{\Gamma(\frac{2j+2\alpha+k}{2b})} \bigg] \bigg)
\end{equation}
for $v=1,2,3$, $\lambda_{j}$ is given by \eqref{def of aj}, and
\begin{align*}
I_{1} := [1-\epsilon,1-\tfrac{M}{\sqrt{n}}), \qquad I_{2} := [1-\tfrac{M}{\sqrt{n}},1+\tfrac{M}{\sqrt{n}}], \qquad I_{3} := (1+\tfrac{M}{\sqrt{n}},1+\epsilon].
\end{align*}
Equivalently, the above sums can be rewritten using
\begin{align}\label{sums lambda j}
& \sum_{j:\lambda_{j}\in I_{3}} = \sum_{j=j_{-}}^{g_{-}-1}, \qquad \sum_{j:\lambda_{j}\in I_{2}} = \sum_{j= g_{-}}^{g_{+}}, \qquad \sum_{j:\lambda_{j}\in I_{1}} = \sum_{j= g_{+}+1}^{j_{+}},
\end{align}
where $g_{-} := \lceil \frac{bnr^{2b}}{1+\frac{M}{\sqrt{n}}}-\alpha \rceil$, $g_{+} := \lfloor \frac{bnr^{2b}}{1-\frac{M}{\sqrt{n}}}-\alpha \rfloor$. We also define $\theta_{-}^{(n,M)},\theta_{+}^{(n,M)} \in [0,1)$ by
\begin{align}
	\begin{split} \label{def of theta nM}
& \theta_{-}^{(n,M)} := g_{-} - \bigg( \frac{bn r^{2b}}{1+\frac{M}{\sqrt{n}}} - \alpha \bigg) = \bigg\lceil \frac{bn r^{2b}}{1+\frac{M}{\sqrt{n}}} - \alpha \bigg\rceil - \bigg( \frac{bn r^{2b}}{1+\frac{M}{\sqrt{n}}} - \alpha \bigg), \\
& \theta_{+}^{(n,M)} := \bigg( \frac{bn r^{2b}}{1-\frac{M}{\sqrt{n}}} - \alpha \bigg) - g_{+} = \bigg( \frac{bn r^{2b}}{1-\frac{M}{\sqrt{n}}} - \alpha \bigg) - \bigg\lfloor \frac{bn r^{2b}}{1-\frac{M}{\sqrt{n}}} - \alpha \bigg\rfloor.
	\end{split}
\end{align}
Note that the sums $S_{2}^{(1)}$ and $S_{2}^{(3)}$ each contain a number of elements proportional to $n$, while $S_{2}^{(2)}$ contains roughly $M\sqrt{n}$ elements. 
\subsection{Global analysis: large $n$ asymptotics of $S_{2}^{(1)}$ and $S_{2}^{(3)}$}
We first treat $S_{2}^{(1)}$, $S_{2}^{(3)}$. These sums are delicate to analyze because they involve the asymptotics of $\gamma(a,z)$ in the regime $a \to + \infty$, $z \to +\infty$, when $\lambda=\frac{z}{a}$ is close to $1$ but not very close (more precisely, $\lambda \in [1-\epsilon,1-\frac{M}{\sqrt{n}})\cup(1+\frac{M}{\sqrt{n}},1+\epsilon]$).
\begin{lemma}\label{lemma:S2kp1p}
As $n \to + \infty$,
\begin{align*}
& S_{2}^{(1)} = n \int_{br^{2b}}^{\frac{br^{2b}}{1-\epsilon}} a \ln \bigg( \bigg( \frac{x}{b} \bigg)^{\frac{1}{2b}} - r \bigg)dx + \bigg\{ a b r^{2b} M \bigg( \ln \bigg( \frac{2b\sqrt{n}}{r M} \bigg) + 1 \bigg) + \frac{a(a-1)b}{2M} \\
& - \frac{a b (3-5a+2a^{2})}{12r^{2b}M^{3}} \bigg\} \sqrt{n} + \bigg( -abr^{2b}M^{2}+a\frac{2\alpha+2\theta_{+}^{(n,M)}-1}{2} - \frac{a}{4}(a-2b+4\alpha) \bigg) \ln \bigg( \frac{M}{\sqrt{n}} \bigg) \\
& + \frac{ar^{2b}(2b-1+8b\ln(\frac{2b}{r}))}{8}M^{2} + \frac{2\alpha+2\theta_{+}^{(n,M)}-1}{2}a \ln \bigg( \frac{r}{2b} \bigg) \\
& + \frac{1-2\alpha-2\theta_{+}^{(n,\epsilon)}}{2}a \ln \bigg( r(1-\epsilon)^{-\frac{1}{2b}}-r \bigg) + \frac{a}{4} \bigg\{ \frac{(1-a)(1+2b)}{2} + (a-2b+4\alpha)\ln(2b) \\
& + \frac{1-a}{(1-\epsilon)^{-\frac{1}{2b}}-1} + (a-2b+4\alpha)\ln \bigg( (1-\epsilon)^{-\frac{1}{2b}}-1 \bigg) \bigg\} + \bigO\bigg( \frac{\sqrt{n}}{M^{5}} + \frac{M^{3}\ln n}{\sqrt{n}} \bigg).
\end{align*}
\end{lemma}
\begin{remark}
Since $M=n^{\frac{1}{8}}(\ln n)^{-\frac{1}{8}}$, the $\bigO$-term above is small as $n \to + \infty$.
\end{remark}
\begin{remark}
The above asymptotics are of the form
\begin{align*}
S_{2}^{(1)} = E_{1}^{(\epsilon)}n + \tilde{E}_{2}^{(M)} \sqrt{n} \ln n + E_{2}^{(M)}\sqrt{n} + \tilde{E}_{3}^{(n,M)} \ln n + E_{3}^{(n,\epsilon,M)} + o(1).
\end{align*}
This may seem a bit counter intuitive as the asymptotics of Theorem \ref{thm:main thm} only contain terms proportional to $n$, $\sqrt{n}$ and $1$. In fact, remarkable cancellations will occur in the asymptotics of $S_{2}^{(1)}+S_{2}^{(2)}+S_{2}^{(3)}$; in particular $\tilde{E}_{2}^{(M)}$ and $\tilde{E}_{3}^{(n,M)}$ will get perfectly canceled by other terms in the large $n$ asymptotics of $S_{2}^{(2)}$ and $S_{2}^{(3)}$, and we will show in Lemma \ref{lemma: asymp of S2k final} below that the large $n$ asymptotics of $S_{2} = S_{2}^{(1)}+S_{2}^{(2)}+S_{2}^{(3)}$ are of the form $S_{2} =  \widehat{C}_{1}^{(\epsilon)} n + \widehat{C}_{2} \sqrt{n} + \widehat{C}_{3}^{(n,\epsilon)} + o(1)$.  
\end{remark}
\begin{proof}
By \eqref{asymp prelim of S2kpvp} and Lemma \ref{lemma: uniform}, $S_{2}^{(1)}$ admits the following exact formula
\begin{align*}
S_{2}^{(1)} & = \sum_{j:\lambda_{j}\in I_{1}} \ln \bigg(\sum_{k=0}^{a} {a \choose k} \frac{(-r)^{a-k}}{n^{\frac{k}{2b}}} \frac{\Gamma(\frac{2j+2\alpha+k}{2b})}{\Gamma(\frac{2j+2\alpha}{2b})} \\
& \hspace{3cm} \times \bigg[ 1+((-1)^{a}e^{u}-1)\bigg( \frac{1}{2}\mathrm{erfc}\Big(\hspace{-0.1cm}-\hspace{-0.05cm}\eta_{j,k} \sqrt{\tfrac{a_{j,k}}{2}}\Big)  - R_{a_{j,k}}(\eta_{j,k}) \bigg) \bigg]\bigg),
\end{align*}
where $\eta_{j,k}$ and $a_{j,k}$ are given by \eqref{def of ajk}. Since $I_{1}=[1-\epsilon,1-\frac{M}{\sqrt{n}})$, 
\begin{align}
& \eta_{j,k} = \lambda_{j,k}-1 +\bigO((\lambda_{j,k}-1)^{2}) \leq -\tfrac{M}{\sqrt{n}} + \bigO(\tfrac{M^{2}}{n}), & & \mbox{as } n \to +\infty, \label{lol10} \\
& -\eta_{j,k} \sqrt{a_{j,k}/2} =-\eta_{j,k} \sqrt{ \tfrac{n r^{2b}}{ 2\lambda_{j,k} } } \geq  \tfrac{Mr^{b}}{\sqrt{2}} + \bigO(\tfrac{M^{2}}{\sqrt{n}}), & & \mbox{as } n \to +\infty, \label{lol11}
\end{align}
uniformly for $j\in \{j: \lambda_{j} \in I_{1}\}$ and $k \in \{0,1,\ldots,a\}$. Also, $M = n^{\frac{1}{8}}(\ln n)^{-\frac{1}{8}}$, and thus, by \eqref{asymp of Ra}, 
\begin{align*}
& R_{a_{j,k}}(\eta_{j,k}) =  \bigO(e^{-\frac{r^{2b}M^{2}}{ 4 }}) = \bigO(e^{-n^{c}}),  & & \frac{1}{2}\mathrm{erfc}\Big(\hspace{-0.1cm}-\hspace{-0.05cm}\eta_{j,k} \sqrt{\tfrac{a_{j,k}}{2}}\Big)  = \bigO(e^{-n^{c}}),
\end{align*}
as $n \to + \infty$ uniformly for $j\in \{j: \lambda_{j} \in I_{1}\}$ and $k \in \{0,1,\ldots,a\}$, and (using also \eqref{sums lambda j})
\begin{align}
S_{2}^{(1)} & = \sum_{j=g_{+}+1}^{j_{+}} \ln \bigg(\sum_{k=0}^{a} {a \choose k} \frac{(-r)^{a-k}}{n^{\frac{k}{2b}}} \frac{\Gamma(\frac{2j+2\alpha+k}{2b})}{\Gamma(\frac{2j+2\alpha}{2b})} \bigg) + \bigO(e^{-n^{c}}), \qquad \mbox{as } n \to + \infty. \label{lol12} 
\end{align}
In the same way as for \eqref{lol1}, we obtain
\begin{align}
& \sum_{k=0}^{a} {a \choose k} \frac{(-r)^{a-k}}{n^{\frac{k}{2b}}} \frac{\Gamma(\frac{2j+2\alpha+k}{2b})}{\Gamma(\frac{2j+2\alpha}{2b})} \sim \gamma_{0}(j/n) + \sum_{\ell=1}^{+\infty} \frac{\sum_{m=1}^{2\ell}\mathfrak{p}_{2\ell,m}\gamma_{m}(j/n)}{j^{\ell}} \label{lol36}
\end{align}
as $n \to + \infty$ uniformly for $j \in \{g_{+}+1,\ldots,j_{+}\}$. However, unlike for $S_{3}$, the first subleading term (corresponding to $\ell=1$) is not sufficient for our purpose; we also need the coefficients of $\mathfrak{p}_{4}$ which are given by (see \eqref{def of p2l})
\begin{align}
 \mathfrak{p}_{4}(k) & = \frac{k(k-2b)}{8}B_2^{(1+\frac{k}{2b})}\Big( \frac{2\alpha+k}{2b}\Big) = \frac{k(k-2b)(8b^{2}+3(k+4\alpha)^{2}-2b(7k+24\alpha))}{384b^{2}}. \label{def of pfrak 4}
\end{align}
The reason as to why we need $\mathfrak{p}_{4}$ can be seen as follows. The sum on the left-hand side of \eqref{lol36} appears inside $\ln$ in \eqref{lol12}. Since $\ln(\gamma_{0}+B)=\ln \gamma_{0} + \ln (1+B/\gamma_{0})$, what is relevant is to estimate the asymptotics series of \eqref{lol36} divided by $\gamma_{0}(j/n)$. Lemma \ref{lemma:gammal} implies 
\begin{equation}\label{asymp of gamma ratio}
\frac{\gamma_{\ell}(x)}{\gamma_{0}(x)} = \bigO(\min\{|x-br^{2b}|^{-a},|x-br^{2b}|^{-\ell}\}) = \bigO(|x-br^{2b}|^{-\min\{a,\ell\}}), \qquad \mbox{as } x \to br^{2b}.
\end{equation}
Using the definitions of $j_{+}$ and $g_{+}$ (see \eqref{def of jminus and jplus} and \eqref{sums lambda j}), we see that for $j \in \{g_{+}+1,\ldots,j_{+}\}$, $j/n$ lies in $(br^{2b},1]$ and $g_{+}/n-br^{2b}$ is of order $\frac{M}{\sqrt{n}}$ as $n \to +\infty$.
Therefore, for $\ell=1,2,\ldots$, 
\begin{align}\label{lol13}
\sum_{j=g_{+}+1}^{j_{+}}\frac{\sum_{m=1}^{2\ell}\mathfrak{p}_{2\ell,m}\gamma_{m}(j/n)}{j^{\ell}\gamma_{0}(j/n)} = \bigO\bigg( \sum_{j=g_{+}+1}^{j_{+}}\frac{|j/n - br^{2b}|^{-\min\{a,2\ell\}}}{n^{\ell}} \bigg) = \bigO\bigg( \frac{\sqrt{n}}{M^{2\ell}} \bigg),
\end{align}
as $n \to + \infty$ uniformly for $j \in \{g_{+}+1,\ldots,j_{+}\}$. Since $M=n^{\frac{1}{8}}(\ln n)^{-\frac{1}{8}}$, only the terms corresponding to $\ell \geq 3$ in \eqref{lol36} (i.e. the terms associated with $\mathfrak{p}_{6},\mathfrak{p}_{8},\ldots$) will give an error in the asymptotics of $S_{2}^{(1)}$. Substituting \eqref{lol36}--\eqref{def of pfrak 4} in \eqref{lol12} and using \eqref{def of gammaj Stirling} (with $\ell=0,1,2,3,4)$, \eqref{estimate for gamma ell} and \eqref{lol13}, we obtain 
\begin{align}\label{lol19}
& S_{2}^{(1)} = \Sigma_{0}^{(1)} + \frac{1}{n}\Sigma_{1}^{(1)} + \frac{1}{n^{2}}\Sigma_{2}^{(1)} +  \bigO\bigg( \frac{\sqrt{n}}{M^{5}} \bigg)
\end{align}
as $n \to + \infty$, where
\begin{align*}
\Sigma_{\ell}^{(1)} := \sum_{j=g_{+}+1}^{j_{+}} f_{\ell}(j/n), \qquad \ell=0,1,2,
\end{align*}
$f_{0}$ and $f_{1}$ are as in \eqref{def of f0 and f1}, and $f_{2}$ is defined by
\begin{align}
 f_{2}(x) 
& : = \frac{1}{384b^{2}x^{2}}\bigg\{ -16b(b^{2}-6b\alpha+6\alpha^{2})\frac{a\big(\frac{x}{b}\big)^{\frac{1}{2b}}}{(\frac{x}{b}\big)^{\frac{1}{2b}}-r} + 12(3b^{2}-8b\alpha+4\alpha^{2}) \frac{a\big(\frac{x}{b}\big)^{\frac{1}{2b}}\big(a\big(\frac{x}{b}\big)^{\frac{1}{2b}}-r\big)}{\big(a\big(\frac{x}{b}\big)^{\frac{1}{2b}}-r\big)^{2}} \nonumber \\
& +4(6\alpha-5b) \frac{a\big(\frac{x}{b}\big)^{\frac{1}{2b}}\big(a^{2}\big(\frac{x}{b}\big)^{\frac{1}{b}}+(1-3a)r\big(\frac{x}{b}\big)^{\frac{1}{2b}}+r^{2}\big)}{\big(a\big(\frac{x}{b}\big)^{\frac{1}{2b}}-r\big)^{3}} \nonumber \\
& + \frac{3a\big(\frac{x}{b}\big)^{\frac{1}{2b}}    \big[ a^{3}\big(\frac{x}{b}\big)^{\frac{3}{2b}} + (4a-6a^{2}-1)r \big(\frac{x}{b}\big)^{\frac{1}{b}} + (7a-4)r^{2}\big(\frac{x}{b}\big)^{\frac{1}{2b}} -r^{3} \big]}{\big(a\big(\frac{x}{b}\big)^{\frac{1}{2b}}-r\big)^{4}} \bigg\} - \frac{f_{1}(x)^{2}}{2}. \label{def of f2}
\end{align}
Using Lemma \ref{lemma:Riemann sum NEW} (with $A=\frac{br^{2b}}{1-\frac{M}{\sqrt{n}}}$, $a_{0}=1-\alpha-\theta_{+}^{(n,M)}$, $B = \frac{br^{2b}}{1-\epsilon}$ and $b_{0}=-\alpha-\theta_{+}^{(n,\epsilon)}$), we get
\begin{align}
& \Sigma_{0}^{(1)} = n \int_{\frac{br^{2b}}{1-\frac{M}{\sqrt{n}}}}^{\frac{br^{2b}}{1-\epsilon}}f_{0}(x)dx + \frac{(2\alpha+2\theta_{+}^{(n,M)}-1)f_{0}(\frac{br^{2b}}{1-\frac{M}{\sqrt{n}}})+(1-2\alpha-2\theta_{+}^{(n,\epsilon)})f_{0}(\frac{br^{2b}}{1-\epsilon})}{2} + \bigO\bigg( \frac{1}{M\sqrt{n}} \bigg), \nonumber \\
& \frac{1}{n}\Sigma_{1}^{(1)} = \int_{\frac{br^{2b}}{1-\frac{M}{\sqrt{n}}}}^{\frac{br^{2b}}{1-\epsilon}}f_{1}(x)dx + \bigO\bigg( \frac{1}{M^{2}} \bigg), \qquad \frac{1}{n^{2}}\Sigma_{2}^{(1)} = \frac{1}{n}\int_{\frac{br^{2b}}{1-\frac{M}{\sqrt{n}}}}^{\frac{br^{2b}}{1-\epsilon}}f_{2}(x)dx + \bigO\bigg( \frac{1}{M^{4}} \bigg), \label{lol14}
\end{align}
as $n \to + \infty$. To obtain the above error terms, we also used \eqref{asymp of gamma ratio}, \eqref{def of gammaj Stirling}, and in particular that
\begin{align*}
f_{0}'\bigg( \frac{br^{2b}}{1-\frac{M}{\sqrt{n}}} \bigg)= \bigO\bigg( \frac{\sqrt{n}}{M} \bigg), \qquad f_{1}\bigg( \frac{br^{2b}}{1-\frac{M}{\sqrt{n}}} \bigg) = \bigO\bigg( \frac{n}{M^{2}} \bigg), \qquad f_{2}\bigg( \frac{br^{2b}}{1-\frac{M}{\sqrt{n}}} \bigg) = \bigO\bigg( \frac{n^{2}}{M^{4}} \bigg)
\end{align*}
as $n \to + \infty$. Furthermore, it follows from a long but straightforward analysis of $f_{0}$, $f_{1}$ and $f_{2}$ that
\begin{align}
f_{0}\bigg( \frac{br^{2b}}{1-\frac{M}{\sqrt{n}}} \bigg) & = a \ln \bigg( \frac{Mr}{2b\sqrt{n}} \bigg) + \bigO \bigg( \frac{M}{\sqrt{n}} \bigg), \label{lol15} \\
n \int_{\frac{br^{2b}}{1-\frac{M}{\sqrt{n}}}}^{\frac{br^{2b}}{1-\epsilon}}f_{0}(x)dx & = n \int_{br^{2b}}^{\frac{br^{2b}}{1-\epsilon}}f_{0}(x)dx + \sqrt{n} \; a b  r^{2b} M \bigg( 1+\ln\bigg( \frac{2b\sqrt{n}}{r M} \bigg) \bigg) + a b r^{2b} M^{2} \ln \bigg( \frac{\sqrt{n}}{M} \bigg) \nonumber \\
& \quad + \frac{a r^{2b} M^{2}(-1+2b+8b \ln (\frac{2b}{r}))}{8} + \bigO \bigg( \frac{M^{3}\ln n}{\sqrt{n}} \bigg), \label{lol16} \\
\int_{\frac{br^{2b}}{1-\frac{M}{\sqrt{n}}}}^{\frac{br^{2b}}{1-\epsilon}}f_{1}(x)dx & = \frac{a(a-1)b\sqrt{n}}{2M} + \frac{a}{4}\bigg\{ \frac{(1-a)(1+2b)}{2} + \frac{1-a}{(\frac{1}{1-\epsilon})^{\frac{1}{2b}}-1} - (a-2b+4\alpha)\ln \bigg( \frac{M}{2b\sqrt{n}} \bigg) \nonumber \\
& \quad + (a-2b+4\alpha) \ln \bigg( (1-\epsilon)^{-\frac{1}{2b}}-1 \bigg) \bigg\} + \bigO\bigg( \frac{M}{\sqrt{n}} \bigg), \label{lol17} \\
\frac{1}{n}\int_{\frac{br^{2b}}{1-\frac{M}{\sqrt{n}}}}^{\frac{br^{2b}}{1-\epsilon}}f_{2}(x)dx & = - \frac{a b (3-5a+2a^{2})\sqrt{n}}{12 r^{2b} M^{3}} + \bigO\bigg( \frac{1}{M^{2}} \bigg), \label{lol18}
\end{align}
as $n \to + \infty$. Substituting \eqref{lol15}--\eqref{lol18} in \eqref{lol14} and then in \eqref{lol19}, we obtain the claim after another long but direct computation.
\end{proof}

\begin{lemma}\label{lemma:S2kp3p}
As $n \to + \infty$,
\begin{align*}
S_{2}^{(3)} = & \; n \int_{\frac{br^{2b}}{1+\epsilon}}^{br^{2b}}\bigg[ u + a \ln \bigg( r - \bigg( \frac{x}{b} \bigg)^{\frac{1}{2b}} \bigg) \bigg] dx + \bigg\{ b r^{2b} M \bigg( a\ln \bigg( \frac{2b\sqrt{n}}{r M} \bigg) + a - u \bigg) + \frac{a(a-1)b}{2M} \\
& - \frac{a b (3-5a+2a^{2})}{12r^{2b}M^{3}} \bigg\} \sqrt{n} + \bigg( abr^{2b}M^{2}+a\frac{2\theta_{-}^{(n,M)}-1-2\alpha}{2} + \frac{a}{4}(a-2b+4\alpha) \bigg) \ln \bigg( \frac{M}{\sqrt{n}} \bigg) \\
& + \frac{r^{2b}(a-2ab+8bu+8ab\ln(\frac{r}{2b}))}{8}M^{2} + \frac{2\theta_{-}^{(n,M)}-1-2\alpha}{2}\bigg( u + a \ln \bigg( \frac{r}{2b} \bigg)\bigg) \\
& + \frac{1+2\alpha-2\theta_{-}^{(n,\epsilon)}}{2}\bigg( u + a \ln \bigg( r-r(1+\epsilon)^{-\frac{1}{2b}} \bigg) \bigg) + \frac{a}{4} \bigg\{ \frac{(a-1)(1+2b)}{2} - (a-2b+4\alpha)\ln(2b) \\
& + \frac{1-a}{1-(1+\epsilon)^{-\frac{1}{2b}}} - (a-2b+4\alpha)\ln \bigg( 1-(1+\epsilon)^{-\frac{1}{2b}} \bigg) \bigg\} + \bigO\bigg( \frac{\sqrt{n}}{M^{5}} + \frac{M^{3}\ln n}{\sqrt{n}} \bigg).
\end{align*}
\end{lemma}
\begin{proof}
By \eqref{asymp prelim of S2kpvp} and Lemma \ref{lemma: uniform},
\begin{align*}
S_{2}^{(3)} & = \sum_{j:\lambda_{j}\in I_{3}} \ln \bigg(\sum_{k=0}^{a} {a \choose k} \frac{(-r)^{a-k}}{n^{\frac{k}{2b}}} \frac{\Gamma(\frac{2j+2\alpha+k}{2b})}{\Gamma(\frac{2j+2\alpha}{2b})} \\
& \hspace{3cm} \times \bigg[ 1+((-1)^{a}e^{u}-1)\bigg( \frac{1}{2}\mathrm{erfc}\Big(\hspace{-0.1cm}-\hspace{-0.05cm}\eta_{j,k} \sqrt{\tfrac{a_{j,k}}{2}}\Big) - R_{a_{j,k}}(\eta_{j,k}) \bigg) \bigg]\bigg).
\end{align*}
Because $I_{3} = (1+\frac{M}{\sqrt{n}},1+\epsilon]$, we have
\begin{align*}
& \eta_{j,k} = \lambda_{j,k}-1 +\bigO((\lambda_{j,k}-1)^{2}) \geq \tfrac{M}{\sqrt{n}} + \bigO(\tfrac{M^{2}}{n}), & & \mbox{as } n \to \infty, \\
& -\eta_{j,k} \sqrt{a_{j,k}/2} \leq - \tfrac{M r^{b}}{\sqrt{2}} + \bigO(\tfrac{M^{2}}{\sqrt{n}}), & & \mbox{as } n \to \infty,
\end{align*}
uniformly for $j\in \{j:\lambda_{j}\in I_{3}\}$ and $k \in \{1,\ldots,a\}$. Since $M = n^{\frac{1}{8}}(\ln n)^{-\frac{1}{8}}$, we get
\begin{align*}
& R_{a_{j,k}}(\eta_{j,k}) = \bigO(e^{-\frac{r^{2b}M^{2}}{4}  }) = \bigO(e^{-n^{c}}), & & \frac{1}{2}\mathrm{erfc}\Big(\hspace{-0.1cm}-\hspace{-0.05cm}\eta_{j,k} \sqrt{\tfrac{a_{j,k}}{2}}\Big)  = 1-\bigO(e^{-\frac{r^{2b}M^{2}}{4} }) = 1-\bigO(e^{-n^{c}}),
\end{align*}
as $n \to + \infty$ uniformly for $j\in \{j:\lambda_{j}\in I_{3}\}$ and $k \in \{1,\ldots,a\}$, and thus 
\begin{align}\label{lol20}
S_{2}^{(3)} & = \sum_{j=j_{-}}^{g_{-}-1} \ln \bigg( \sum_{k=0}^{a} {a \choose k} \frac{r^{a-k}(-1)^{k}e^{u}}{n^{\frac{k}{2b}}} \frac{\Gamma(\frac{2j+2\alpha+k}{2b})}{\Gamma(\frac{2j+2\alpha}{2b})} \bigg) + \bigO(e^{-n^{c}}),
\end{align}
where we have also used \eqref{sums lambda j}. By \eqref{lol37},
\begin{align}\label{lol37 bis}
& \sum_{k=0}^{a} {a \choose k} \frac{r^{a-k}(-1)^{k}e^{u}}{n^{\frac{k}{2b}}} \frac{\Gamma(\frac{2j+2\alpha+k}{2b})}{\Gamma(\frac{2j+2\alpha}{2b})} \sim e^{u} \bigg\{\gamma_{0}(j/n) + \sum_{\ell=1}^{+\infty} \frac{\sum_{m=1}^{2\ell}\mathfrak{p}_{2\ell,m}\gamma_{m}(j/n)}{j^{\ell}} \bigg\},
\end{align}
as $n \to + \infty$ uniformly for $j \in \{j_{-},\ldots,g_{-}-1\}$. Using Lemma \ref{lemma:gammal}, \eqref{asymp of gamma ratio} and the definitions \eqref{def of jminus and jplus}, \eqref{sums lambda j} of $j_-$ and $g_-$, we infer that \eqref{lol13} holds also as $n \to + \infty$ uniformly for $j\in \{j_{-},\ldots,g_{-}-1\}$. Hence, substituting \eqref{lol37 bis} in \eqref{lol20} and using \eqref{def of gammaj Stirling} (with $\ell=0,1,2,3,4$), \eqref{estimate for gamma ell} and \eqref{lol13}, we get
\begin{align}\label{lol21}
& S_{2}^{(3)} = \Sigma_{0}^{(3)} + \frac{1}{n}\Sigma_{1}^{(3)} + \frac{1}{n^{2}}\Sigma_{2}^{(3)} +  \bigO\bigg( \frac{\sqrt{n}}{M^{5}} \bigg)
\end{align}
as $n \to + \infty$, where
\begin{align*}
\Sigma_{0}^{(3)} := \sum_{j=j_{-}}^{g_{-}-1} \widetilde{f}_{0}(j/n), \qquad \Sigma_{\ell}^{(3)} := \sum_{j=j_{-}}^{g_{-}-1} f_{\ell}(j/n), \qquad \ell=1,2,
\end{align*}
$\widetilde{f}_{0}$ is given by \eqref{lol22}, $f_{1}$ is given by \eqref{def of f0 and f1}, and $f_{2}$ is given by \eqref{def of f2}. 

Using Lemma \ref{lemma:Riemann sum NEW} (with $A=\frac{br^{2b}}{1+\epsilon}$, $a_{0}=\theta_{-}^{(n,\epsilon)}-\alpha$, $B = \frac{br^{2b}}{1+\frac{M}{\sqrt{n}}}$ and $b_{0} = \theta_{-}^{(n,M)}-1-\alpha$), we get
\begin{align}
& \Sigma_{0}^{(3)} = n \int_{\frac{br^{2b}}{1+\epsilon}}^{\frac{br^{2b}}{1+\frac{M}{\sqrt{n}}}}\widetilde{f}_{0}(x)dx + \frac{(1+2\alpha-2\theta_{-}^{(n,\epsilon)})\widetilde{f}_{0}(\frac{br^{2b}}{1+\epsilon})+(2\theta_{-}^{(n,M)}-1-2\alpha)\widetilde{f}_{0}(\frac{br^{2b}}{1+\frac{M}{\sqrt{n}}})}{2} + \bigO\bigg( \frac{1}{M\sqrt{n}} \bigg), \nonumber \\
& \frac{1}{n}\Sigma_{1}^{(3)} = \int_{\frac{br^{2b}}{1+\epsilon}}^{\frac{br^{2b}}{1+\frac{M}{\sqrt{n}}}} f_{1}(x)dx + \bigO\bigg( \frac{1}{M^{2}} \bigg), \qquad \frac{1}{n^{2}}\Sigma_{2}^{(3)} = \frac{1}{n}\int_{\frac{br^{2b}}{1+\epsilon}}^{\frac{br^{2b}}{1+\frac{M}{\sqrt{n}}}} f_{2}(x)dx + \bigO\bigg( \frac{1}{M^{4}} \bigg), \label{lol24}
\end{align}
as $n \to + \infty$. Furthermore, a long but straightforward analysis of $\widetilde{f}_{0}$, $f_{1}$ and $f_{2}$ shows that
\begin{align}
\widetilde{f}_{0}\bigg( \frac{br^{2b}}{1+\frac{M}{\sqrt{n}}} \bigg) & = u+ a \ln \bigg( \frac{Mr}{2b\sqrt{n}} \bigg) + \bigO \bigg( \frac{M}{\sqrt{n}} \bigg), \label{lol25} \\
n \int_{\frac{br^{2b}}{1+\epsilon}}^{\frac{br^{2b}}{1+\frac{M}{\sqrt{n}}}}\widetilde{f}_{0}(x)dx & = n \int_{\frac{br^{2b}}{1+\epsilon}}^{br^{2b}}\widetilde{f}_{0}(x)dx + \sqrt{n} \; b  r^{2b} M \bigg( a-u+a\ln\bigg( \frac{2b\sqrt{n}}{r M} \bigg) \bigg) + a b r^{2b} M^{2} \ln \bigg( \frac{M}{\sqrt{n}} \bigg) \nonumber \\
& \quad  + \frac{r^{2b} M^{2}(a-2ab+8bu+8ab \ln (\frac{r}{2b}))}{8} + \bigO \bigg( \frac{M^{3}\ln n}{\sqrt{n}} \bigg), \label{lol26} \\
\int_{\frac{br^{2b}}{1+\epsilon}}^{\frac{br^{2b}}{1+\frac{M}{\sqrt{n}}}}f_{1}(x)dx & = \frac{a(a-1)b\sqrt{n}}{2M} + \frac{a}{4}\bigg\{ \frac{(a-1)(1+2b)}{2} + \frac{1-a}{1-(\frac{1}{1+\epsilon})^{\frac{1}{2b}}} + (a-2b+4\alpha)\ln \bigg( \frac{M}{2b\sqrt{n}} \bigg) \nonumber \\
& \quad - (a-2b+4\alpha) \ln \bigg( 1-(1+\epsilon)^{-\frac{1}{2b}} \bigg) \bigg\} + \bigO\bigg( \frac{M}{\sqrt{n}} \bigg), \label{lol27} \\
\frac{1}{n}\int_{\frac{br^{2b}}{1+\epsilon}}^{\frac{br^{2b}}{1+\frac{M}{\sqrt{n}}}}f_{2}(x)dx & = - \frac{a b (3-5a+2a^{2})\sqrt{n}}{12 r^{2b} M^{3}} + \bigO\bigg( \frac{1}{M^{2}} \bigg). \label{lol28}
\end{align}
Substituting \eqref{lol25}--\eqref{lol28} in \eqref{lol24} and then in \eqref{lol21}, we obtain the claim after another long but direct computation.
\end{proof}
\subsection{Local analysis: large $n$ asymptotics of $S_{2}^{(2)}$}
Our next goal is to obtain the large $n$ asymptotics of $S_{2}^{(2)}$. This is the most technical part of the proof of Theorem \ref{thm:main thm}. As mentioned in the introduction, a major obstacle in the asymptotic analysis of $S_{2}^{(2)}$ is that, in order to treat the general case $a \in \mathbb{N}$, we need to expand various quantities to all orders. Lemma \ref{lemma:some computation for S2kp2p} below provides a general scheme to compute the coefficients appearing in these expansions in a recursive way. These coefficients are not all readily available in explicit forms, however only a few of those will really matter for us. Lemmas \ref{lemma:sum of qaa(1)} and \ref{lemma:sum of qaa(4)} establish some non-trivial identities between those relevant coefficients and the (associated) Hermite polynomials. The large $n$ asymptotics of $S_{2}^{(2)}$ are then obtained in Lemma \ref{lemma:S2kp2p}.

Let us define
\begin{align}
& M_{j,k} := \sqrt{n}(\lambda_{j,k}-1), & & k \in \{1,\ldots,m\}, \; j \in \{j: \lambda_{j} \in I_{2}\}=\{g_{-},\ldots,g_{+}\},  \label{Mjk}
\\
& M_{j} := \sqrt{n}(\lambda_{j}-1), & & j \in \{j: \lambda_{j} \in I_{2}\}=\{g_{-},\ldots,g_{+}\}.  \label{Mj}
\end{align}
Since $I_{2}=[1-\tfrac{M}{\sqrt{n}},1+\tfrac{M}{\sqrt{n}}]$ and $\lambda_{j}$ is decreasing in $j$, we have $-M \leq M_{g_{+}} < \ldots < M_{g_{-}} \leq M$. Note from \eqref{def of aj}--\eqref{def of ajk} that $\lambda_{j,k}$ is close to $\lambda_{j}$ for large $n$, and from \eqref{Mjk}--\eqref{Mj} that $M_{j,k}$ is close to $M_{j}$ for large $n$. For convenience, for $j \in \{g_{-},\ldots,g_{+}\}$, we also define
\begin{align}\label{def of Xijk}
\Xi_{j,k} := \frac{M_j r^b}{\sqrt{2}} - \frac{\sqrt{a_{j,k}}}{\sqrt{2}} \eta_{j,k},
\end{align}
where $a_{j,k}$ and $\eta_{j,k}$ are given in \eqref{def of ajk}. 
\begin{lemma}\label{lemma:some computation for S2kp2p}
Let $k \in \{0,1,\ldots,a\}$ be fixed. As $n \to + \infty$ and uniformly for $j \in \{g_{-},\ldots,g_{+}\}$, we have
\begin{align}
& \frac{1}{n^{\frac{k}{2b}}}\frac{\Gamma(\frac{2j+2\alpha+k}{2b})}{\Gamma(\frac{2j+2\alpha}{2b})} \sim r^{k} \sum_{\ell=0}^{+\infty} \frac{\sum_{m=0}^{\ell}\sum_{p=0}^{\ell} \mathfrak{q}_{\ell,m,p}^{(1)}k^{m}M_{j}^{p}}{n^{\frac{\ell}{2}}},   \label{Gamma ratio qlmp(1)}
\\
& \Xi_{j,k} \sim \Xi_{j,k}^{\mathrm{formal}} := \sum_{\ell=1}^{+\infty} \frac{\sum_{m=0}^{ \lfloor (\ell+1)/2 \rfloor } d_{\ell,m} \, k^m M_j^{\ell+1-2m}}{ n^{\ell/2} }, \label{Xi expansion} \\
& \frac{1}{2}\mathrm{erfc}\Big(\hspace{-0.1cm}-\hspace{-0.05cm}\eta_{j,k} \sqrt{\tfrac{a_{j,k}}{2}}\Big) \sim \frac{1}{2}\mathrm{erfc}\Big( \hspace{-0.1cm}-\hspace{-0.05cm}\frac{M_{j}r^{b}}{\sqrt{2}} \Big)-\frac{e^{-\frac{M_{j}^{2}r^{2b}}{2}}}{\sqrt{2\pi}} \sum_{\ell=1}^{+\infty} \frac{\sum_{m=0}^{\ell}\sum_{p=0}^{3\ell-1-2m} \mathfrak{q}_{\ell,m,p}^{(2)}k^{m}M_{j}^{p}}{n^{\frac{\ell}{2}}}, \label{erfc qlmp(2)}
\\
& R_{a_{j,k}}(\eta_{j,k}) \sim -\frac{e^{-\frac{M_{j}^{2}r^{2b}}{2}}}{3\sqrt{2\pi} \; r^{b} \sqrt{n}} \sum_{\ell=0}^{+\infty} \frac{\sum_{m=0}^{\ell}\sum_{p=0}^{3\ell-2m} \mathfrak{q}_{\ell,m,p}^{(3)}k^{m}M_{j}^{p}}{n^{\frac{\ell}{2}}}, \label{R qlmp(3)} \\
&\frac{1}{n^{\frac{k}{2b}}} \frac{\Gamma(\frac{2j+2\alpha+k}{2b})}{\Gamma(\frac{2j+2\alpha}{2b})} \bigg[ 1+((-1)^{a}e^{u}-1)\bigg( \frac{1}{2}\mathrm{erfc}\Big(\hspace{-0.1cm}-\hspace{-0.05cm}\eta_{j,k} \sqrt{\tfrac{a_{j,k}}{2}}\Big)  - R_{a_{j,k}}(\eta_{j,k}) \bigg) \bigg] \sim r^{k}\sum_{\ell=0}^{+\infty} \frac{\mathcal{A}_{\ell}(M_{j};k)}{\sqrt{n^{\ell} }}, \label{lol29}
\end{align}
for some $\mathfrak{q}_{\ell,m,p}^{(1)},\mathfrak{q}_{\ell,m,p}^{(2)},\mathfrak{q}_{\ell,m,p}^{(3)},\mathfrak{q}_{\ell,m,p}^{(4)},d_{\ell,m} \in \mathbb{C},$ $\mathfrak{q}_{0,0,0}^{(1)}:=1$, $\mathfrak{q}_{0,0,0}^{(3)}:=1$, where $a_{j,k}$, $\eta_{j,k}$ are given in \eqref{def of ajk}, and the $\mathcal{A}_{\ell}$'s are defined by
\begin{align}
\mathcal{A}_{0}(x;k) & =  1+\frac{(-1)^{a}e^{u}-1}{2}\mathrm{erfc}\bigg( -\frac{r^{b}x}{\sqrt{2}} \bigg), \label{def of Acal0} \\
\mathcal{A}_{\ell}(x;k) & = \bigg( 1+\frac{(-1)^{a}e^{u}-1}{2}\mathrm{erfc}\bigg( -\frac{r^{b}x}{\sqrt{2}} \bigg) \bigg) \sum_{m=1}^{\ell}\sum_{p=0}^{\ell}\mathfrak{q}_{\ell,m,p}^{(1)}k^{m}x^{p} \nonumber \\
&\quad  + ((-1)^{a}e^{u}-1)\frac{e^{-\frac{r^{2b}x^{2}}{2}}}{\sqrt{2\pi} } \sum_{m=0}^{\ell}\sum_{p=0}^{ 3\ell-1-2m } \mathfrak{q}_{\ell,m,p}^{(4)}k^{m}x^{p}, \quad \ell \geq 1. \label{expansion of mathcalA}
\end{align}
For $\ell \geq 1$, $m=0$ and $p=0,\ldots,\ell$, $\mathfrak{q}_{\ell,m,p}^{(1)}:=0$. The other coefficients $\{\mathfrak{q}_{\ell,m,p}^{(1)}\}$ are given by
\begin{align}
& \sum_{m=1}^{\ell}\sum_{p=0}^{\ell} \mathfrak{q}_{\ell,m,p}^{(1)}k^{m}x^{p}=  \sum_{s=0 }^{ \lfloor \ell/2 \rfloor  } \frac{1}{  ( r^{2b} )^{s} }  \binom{\frac{k}{2b}}{s} B_s^{(1+ \frac{k}{2b}) }\Big( \frac{k}{2b}\Big)  \binom{ s-\frac{k}{2b} }{ \ell-2s } x^{\ell-2s}, \label{mathfrak qlmp(1)}
\end{align}
and the coefficients $\{\mathfrak{q}_{\ell,m,p}^{(j)}\}_{j=2}^{4}$ can be found by equaling the terms of the same order in $n$ in the following formal power series
\begin{align}
& \sum_{\ell=1}^{+\infty} \frac{\sum_{m=0}^{\ell}\sum_{p=0}^{3\ell-1-2m} \mathfrak{q}_{\ell,m,p}^{(2)}k^{m}M_{j}^{p}}{n^{\frac{\ell}{2}}} = \sum_{\ell=1}^{+\infty} \frac{ 2^{ \ell/2 } }{\ell!} \He_{\ell-1} ( M_j r^b ) (\Xi_{j,k}^{\mathrm{formal}})^\ell,   \label{mathfrak qlmp(2)} 
\\
& \sum_{\ell=0}^{+\infty} \frac{\sum_{m=0}^{\ell}\sum_{p=0}^{3\ell-2m} \mathfrak{q}_{\ell,m,p}^{(3)}k^{m}M_{j}^{p}}{n^{\frac{\ell}{2}}} = -3 r^{b}\sqrt{n} \bigg( \sum_{\ell=0}^{+\infty} \frac{ 2^{ \ell/2 } }{\ell!}\, \He_\ell( M_j r^b )  (\Xi_{j,k}^{\mathrm{formal}})^\ell  \bigg) \bigg(\sum_{\ell=0}^{+\infty} \frac{c_{\ell}(\eta_{j,k})}{a_{j,k}^{\ell+\frac{1}{2}}} \bigg),   \label{mathfrak qlmp(3)} 
\\
& \sum_{\ell=1}^{+\infty} \frac1{n^{ \frac{\ell}{2} }}  \sum_{m=0}^{\ell}\sum_{p=0}^{3\ell-1-2m} \mathfrak{q}_{\ell,m,p}^{(4)}k^{m}M_{j}^{p}
	= -\bigg( 1+ \sum_{\ell=1}^{+\infty} \frac{\sum_{m=1}^{\ell}\sum_{p=0}^{\ell} \mathfrak{q}_{\ell,m,p}^{(1)}k^{m}M_{j}^{p}}{n^{\frac{\ell}{2}}}\bigg)  \nonumber \\
	& \hspace{0.5cm}\times \bigg\{ \sum_{\ell=1}^{+\infty} \frac{\sum_{m=0}^{\ell}\sum_{p=0}^{ 3\ell-1-2m } \mathfrak{q}_{\ell,m,p}^{(2)}k^{m}M_{j}^{p}}{n^{\frac{\ell}{2}}} -\frac{1}{3r^{b}\sqrt{n}}\sum_{\ell=0}^{+\infty} \frac{\sum_{m=0}^{\ell}\sum_{p=0}^{3\ell-2m} \mathfrak{q}_{\ell,m,p}^{(3)}k^{m}M_{j}^{p}}{n^{\frac{\ell}{2}}} \bigg\}, \label{q lmp 5 2 3}
\end{align}
where the $c_{\ell}$'s are defined recursively as in \eqref{recursive def of the cj}.  
\end{lemma}
\begin{proof}
By \eqref{def of aj} and \eqref{Mj}, we have $j+\alpha= \frac{bnr^{2b}}{1+\frac{M_{j}}{\sqrt{n}}}$. This, by \eqref{Gamma ratio qlmp(1)}, implies that the coefficients $\mathfrak{q}_{\ell,m,p}^{(1)}$ do not depend on $\alpha$. For $\alpha=0$, using \eqref{large j asymp for the ratio of Gamma}, we get
\begin{align*}
\frac{1}{n^{\frac{k}{2b}}} \frac{\Gamma(\frac{2j+k}{2b})}{\Gamma(\frac{2j}{2b})} &\sim   r^k   \sum_{\ell=0}^{+\infty} \frac{1}{  ( r^{2b} )^{\ell} }  \binom{\frac{k}{2b}}{\ell} B_\ell^{(1+ \frac{k}{2b}) }\Big( \frac{k}{2b}\Big) \frac1{n^{\ell}}\bigg(  1+\frac{M_j}{\sqrt{n}} \bigg)^{\ell-\frac{k}{2b}} \\
& = r^k   \sum_{\ell=0}^{+\infty} \frac{1}{  ( r^{2b} )^{\ell} }  \binom{\frac{k}{2b}}{\ell} B_\ell^{(1+ \frac{k}{2b}) }\Big( \frac{k}{2b}\Big) \frac1{n^{\ell}} \sum_{s=0}^{+\infty} \binom{\ell-\frac{k}{2b}}{s} \frac{M_j^{s}}{n^{\frac{s}{2}}} \\
&= r^k   \sum_{\ell=0}^{+\infty} \bigg[ \sum_{s=0 }^{ \lfloor \ell/2 \rfloor  } \frac{1}{  ( r^{2b} )^{s} }  \binom{\frac{k}{2b}}{s} B_s^{(1+ \frac{k}{2b}) }\Big( \frac{k}{2b}\Big)  \binom{ s-\frac{k}{2b} }{ \ell-2s } M_j^{\ell-2s} \bigg] \frac1{n^{\frac{\ell}{2}}},
\end{align*}
and \eqref{Gamma ratio qlmp(1)}, \eqref{mathfrak qlmp(1)} follow. By \eqref{def of Xijk},
\begin{align} \label{Xi jk lambda jk}
\Xi_{j,k}-\frac{ M_j r^b }{\sqrt{2}} = - \sqrt{nr^{2b}} (\lambda_{j,k}-1)	 \sqrt{ \frac{ \lambda_{j,k}-1-\ln \lambda_{j,k} }{ \lambda_{j,k} (\lambda_{j,k}-1)^{2}   } } \sim - \sqrt{nr^{2b}} \sum_{s=1}^{+\infty} \mathfrak{a}_{s}^{(1)} (\lambda_{j,k}-1)^{s}
\end{align}
as $n \to + \infty$ uniformly for $j \in \{g_{-},\ldots,g_{+}\}$, for some $\{\mathfrak{a}_{s}^{(1)}\}_{s=1}^{+\infty}\subset \mathbb{C}$ that are independent of $j$ and $k$. The all-order expansion \eqref{Xi expansion} now follows from 
\begin{align*}
\lambda_{j,k}&= \Big( 1+\frac{M_j}{\sqrt{n}} \Big)\bigg( 1+\frac{k}{2bnr^{2b}} \Big(1+\frac{M_j}{\sqrt{n}}\Big)\bigg)^{-1} \sim \sum_{\ell=0}^{+\infty} \Big( \frac{-k}{2bnr^{2b}} \Big)^\ell \Big(1+\frac{M_j}{\sqrt{n}}\Big)^{\ell+1}= \sum_{\ell=0}^{+\infty} \frac{ [\lambda_{j,k}]_\ell }{n^{\ell/2}},
\end{align*} 
where $[\lambda_{j,k}]_\ell $ is given by 
\begin{align}\label{lambda jk all order}
[\lambda_{j,k}]_\ell = \sum_{s=\lceil (\ell-1)/3 \rceil }^{\lfloor \ell/2 \rfloor } \Big( \frac{-1}{2br^{2b}} \Big)^s \binom{s+1}{ \ell-2s } \, k^s M_j^{\ell-2s} .
\end{align}
Using that 
\begin{equation} \label{erfc derivatives}
\frac{1}{2}\frac{d^{\ell}}{dz^{\ell}} \mathrm{erfc}(z)= -\frac{1}{\sqrt{\pi}}  \frac{d^{\ell-1}}{dz^{\ell-1}} e^{-z^2} =  \frac{(-1)^{\ell}}{\sqrt{2\pi}}  \, 2^{\frac{\ell}{2}} \He_{\ell-1}(\sqrt{2}z) \,e^{-z^2}, \qquad \ell \in \mathbb{N},
\end{equation}
and $\He_\ell(z)=(-1)^{\ell}\He_\ell(-z)$, we infer that
\begin{align*}
\frac12 \mathrm{erfc} \Big(  -\frac{ M_j r^b }{\sqrt{2}}+ \Xi_{j,k}\Big) \sim \frac12 \mathrm{erfc} \Big(  -\frac{ M_j r^b }{\sqrt{2}}\Big)-\frac{e^{ -\frac{M_j^2 r^{2b}}{2} }}{\sqrt{2\pi}}  \sum_{\ell=1}^{+\infty} \frac{ 2^{ \ell/2 } }{\ell!} \He_{\ell-1} ( M_j r^b ) \Xi_{j,k}^{\ell}
\end{align*}
as $n \to + \infty$ uniformly for $j \in \{g_{-},\ldots,g_{+}\}$. Also, from \eqref{Xi expansion}, we obtain
\begin{align}\label{Xi l expansion}
\Xi_{j,k}^{\ell}\sim (\Xi_{j,k}^{\mathrm{formal}})^{\ell} = \bigg(\sum_{s=1}^{+\infty} \frac{\sum_{m=0}^{ \lfloor (s+1)/2 \rfloor } d_{s,m} \, k^m M_j^{s+1-2m}}{ n^{s/2} } \bigg)^{\ell} = \sum_{s=\ell}^{+\infty} \frac{\sum_{m=0}^{ \lfloor (s+\ell)/2 \rfloor } d_{s,m}^{(\ell)} \, k^m M_j^{s+\ell-2m}}{ n^{s/2} }
\end{align}
for some coefficients $\{d_{s,m}^{(\ell)}\}  \subset \mathbb{C}$. This implies \eqref{erfc qlmp(2)} and \eqref{mathfrak qlmp(2)}. 
Next, by \eqref{asymp of Ra}, we have
\begin{align*}
	& R_{a_{j,k}}(\eta_{j,k}) \sim \frac{e^{-\frac{1}{2}a_{j,k} \eta_{j,k}^{2}}}{\sqrt{2\pi a_{j,k} }}\sum_{\ell=0}^{+\infty} \frac{c_{\ell}(\eta_{j,k})}{a_{j,k}^{\ell}}, \qquad \mbox{as } n \to + \infty,
\end{align*}
uniformly for $j \in \{g_{-},\ldots,g_{+}\}$. Using again \eqref{erfc derivatives}, we get
\begin{align*}
e^{-\frac{1}{2}a_{j,k} \eta_{j,k}^{2}} \sim e^{-\frac{M_{j}^{2}r^{2b}}{2}}  \sum_{\ell=0}^{+\infty}  2^{ \ell/2 }\, \He_\ell( M_j r^b ) \frac{ \Xi_{j,k}^\ell }{ \ell! },
\end{align*}
and \eqref{R qlmp(3)} and \eqref{mathfrak qlmp(3)} follow from a long but direct analysis of the right-hand side of \eqref{mathfrak qlmp(3)}, using that the $c_{\ell}$'s are smooth, and using the all-order expansions \eqref{Xi l expansion} and  
\begin{align} 
& a_{j,k}= \frac{n r^{2b}}{1+\frac{M_j}{\sqrt{n}} }+\frac{k}{2b} \sim  \frac{k}{2b} + nr^{2b} \sum_{\ell=0}^{+\infty}\frac{(-1)^{\ell}M_{j}^{\ell}}{n^{\ell/2}}, \label{ajk expansion} \\
& \eta_{j,k} \sim   \sum_{s=1}^{+\infty} \mathfrak{b}_{s} (\lambda_{j,k}-1)^{s}= \sum_{s=1}^{+\infty} \mathfrak{b}_{s} \bigg( \sum_{\ell=1}^{+\infty} \frac{ [\lambda_{j,k}]_\ell }{n^{\ell/2}} \bigg)^{s}, \label{etajk expansion}
\end{align}
where the coefficients $\{\mathfrak{b}_{s}\}_{s=1}^{+\infty}\subset \mathbb{C}$ are independent of $j$ and $k$.

Finally, \eqref{lol29}, \eqref{def of Acal0}, \eqref{expansion of mathcalA} and \eqref{q lmp 5 2 3} are direct consequences of \eqref{Gamma ratio qlmp(1)}--\eqref{R qlmp(3)}. 
\end{proof}

\begin{lemma} \label{lemma:sum of qaa(1)}
The following relations hold
\begin{align}
& (2b)^a r^{ab} (-1)^a a! \sum_{p=0}^a \mathfrak{q}_{a,a,p}^{(1)} \Big( \frac{x}{r^b}\Big)^{p} =p_{0,a}(x), \label{q aap 1 sum} \\
& (2b)^{a+1} r^{(a+1)b} (-1)^a a!  \sum_{p=0}^{a+1} \mathfrak{q}_{a+1,a,p}^{(1)}  \Big( \frac{x}{r^b}\Big)^{p} \nonumber \\
& \hspace{3cm} = -ab  \Big(  p_{0,a+1}(x) + (1-3a) p_{0,a-1}(x) + \frac{5}{3}(a-1)(a-2) p_{0,a-3}(x) \Big), \label{q a+1ap 1 sum}
\end{align}
where $p_{0,a}(x)$ is given by \eqref{def of p0a}.
\end{lemma} 
\begin{proof}
It follows from \eqref{mathfrak qlmp(1)} that
\begin{align*}
(2b)^a r^{ab} (-1)^a a! \sum_{p=0}^a \mathfrak{q}_{a,a,p}^{(1)} \Big( \frac{x}{r^b}\Big)^{p} &=(2b)^a (-1)^a    	\sum_{s=0 }^{ \lfloor a/2 \rfloor  }  \Big[\frac{d}{dk}\Big]^a \bigg[ B_s^{(1+ \frac{k}{2b}) }\Big( \frac{k}{2b}\Big) \binom{\frac{k}{2b}}{s}   \binom{ s-\frac{k}{2b} }{ a-2s }  \bigg] \bigg|_{k=0}  x^{a-2s}.
\end{align*}
For each $s \in \{0,\ldots,\lfloor \frac{a}{2} \rfloor\}$, the polynomial $k \mapsto B_s^{(1+ \frac{k}{2b}) }( \frac{k}{2b}) \binom{\frac{k}{2b}}{s}   \binom{ s-\frac{k}{2b} }{ a-2s }$ is of the form $\sum_{\ell=0}^{a} \tilde{c}_{\ell}(k/b)^{a}$, where $\tilde{c}_{0},\ldots,\tilde{c}_{a} \in \mathbb{C}$ are independent of $b$. Thanks to the prefactor $(2b)^a$, the above expression is thus independent of $b$. Replacing $b$ by $\frac{1}{2}$ yields
\begin{align}\label{lol41}
(2b)^a r^{ab} (-1)^a a! \sum_{p=0}^a \mathfrak{q}_{a,a,p}^{(1)} \Big( \frac{x}{r^b}\Big)^{p} &=  (-1)^a    \sum_{s=0 }^{ \lfloor a/2 \rfloor  }  \Big[\frac{d}{dk}\Big]^a \bigg[  B_s^{(1+ k) }( k )  \binom{k}{s} \binom{ s-k }{ a-2s }  \bigg] \bigg|_{k=0}  x^{a-2s}.
\end{align}
By \eqref{def of gen Bernouilli}, 
\begin{align}\label{lol39}
B_s^{(1+ k) }( k ) =  \Big( \frac{d}{dt}\Big)^s \Big[ \Big( \frac{t}{e^t-1}\Big)^{k+1} e^{ kt } \Big]\Big|_{t=0} = \frac{1}{2^s}\Big( k^s-\frac{s(s+5)}{6} k^{s-1} + \dots \Big) 
\end{align}
is a polynomial in $k$ of degree $s$, and
\begin{align}\label{lol40}
\binom{k}{s} \binom{ s-k }{ a-2s } = \frac{ (-1)^{a} }{s! (a-2s)!} \Big( k^{a-s}+\frac{ 7s^2-(6a-3)s+a^2-a}{2} \, k^{a-1-s}+\dots \Big)
\end{align}
is a polynomial in $k$ of degree $a-s$. Combining \eqref{lol39} and \eqref{lol40}, we obtain
\begin{align*}
\Big[\frac{d}{dk}\Big]^a \bigg[  B_s^{(1+ k) }( k )  \binom{k}{s} \binom{ s-k }{ a-2s }  \bigg] \bigg|_{k=0} = \frac{ (-1)^{a} a!}{s! (a-2s)! 2^{s}},
\end{align*}
and now \eqref{q aap 1 sum} follows directly from the right-most expression of $p_{0,a}$ in \eqref{def of p0a}. Now we turn to the proof of \eqref{q a+1ap 1 sum}. In a similar way as \eqref{lol41}, using \eqref{mathfrak qlmp(1)},
\begin{multline*}
(2b)^{a+1} r^{(a+1)b} (-1)^a a!  \sum_{p=0}^{a+1} \mathfrak{q}_{a+1,a,p}^{(1)}  \Big( \frac{x}{r^b}\Big)^{p}
\\
= 2b\, (-1)^a  \sum_{s=0 }^{ \lfloor (a+1)/2 \rfloor  }  \Big[\frac{d}{dk}\Big]^a \bigg[ B_s^{(1+k) }(k) \binom{k}{s}   \binom{ s-k }{ a+1-2s } \bigg] \bigg|_{k=0}   x^{a+1-2s} .  
\end{multline*}
From \eqref{lol39} and \eqref{lol40} with $a$ replaced by $a+1$, we get 
\begin{align*}
\Big[\frac{d}{dk}\Big]^a \bigg[ B_s^{(1+k) }(k) \binom{k}{s}   \binom{ s-k }{ a+1-2s } \bigg] \bigg|_{k=0} = \frac{ (-1)^{a+1} a!}{s! (a+1-2s)! 2^{s}} \frac{20s^{2}-2(9a+7)s+3a(a+1)}{6},
\end{align*}
and thus
\begin{multline*}
(2b)^{a+1} r^{(a+1)b} (-1)^a a!  \sum_{p=0}^{a+1} \mathfrak{q}_{a+1,a,p}^{(1)}  \Big( \frac{x}{r^b}\Big)^{p} = -\frac{b}{3} \sum_{s=0 }^{ \lfloor (a+1)/2 \rfloor  }  \frac{a!\, (20s^2-2(9a+7)s+3a(a+1)) }{ s! (a+1-2s)!}  \frac{ x^{a+1-2s} }{ 2^s }.
\end{multline*}
Finally, the expression \eqref{q a+1ap 1 sum} follows from the manipulation
\begin{align*}
&\quad \sum_{s=0 }^{ \lfloor (a+1)/2 \rfloor  }  \frac{a!\, (20s^2-2(9a+7)s+3a(a+1)) }{ s! (a+1-2s)!}  \frac{ x^{a+1-2s} }{ 2^s }
\\
&= 3a \,p_{0,a+1}(x)+ a! \sum_{s=0 }^{ \lfloor (a-1)/2 \rfloor  }  \frac{ 10s+3(1-3a) }{ s! (a-1-2s)!}  \frac{ x^{a-1-2s} }{ 2^s } 
\\
&= 3a \,p_{0,a+1}(x)+ 3a(1-3a) \,p_{0,a-1}(x) + a! \sum_{s=1 }^{ \lfloor (a-1)/2 \rfloor  }  \frac{ 10 }{ (s-1)! (a-1-2s)!}  \frac{ x^{a-1-2s} }{ 2^s } 
\\
&= 3a \,p_{0,a+1}(x)+ 3a(1-3a) \,p_{0,a-1}(x) + 5 a(a-1)(a-2) \, p_{0,a-3}(x).
\end{align*}
\end{proof}

\begin{lemma} \label{lemma:sum of qaa(4)}
The following relations hold
\begin{align}
& (2br^{b})^{a+1} (-1)^{a+1} (a+1)! \sum_{p=0}^{a} \mathfrak{q}_{a+1,a+1,p}^{(4)}  \Big(\frac{x}{r^b}\Big)^{p} = q_{0,a+1}(x), \label{q aap 4 sum} \\
& (2br^{b})^{a+1} (-1)^a a!  \sum_{p=0}^{a+2} \mathfrak{q}_{a+1,a,p}^{(4)}  \Big( \frac{x}{r^b}\Big)^{p} \nonumber \\
& \hspace{3cm} = -b  \Big( a \, q_{0,a+1}(x) + (1-3a) [a q_{0,a-1}(x)] + \frac{5}{3} [a(a-1)(a-2) q_{0,a-3}(x)] \Big), \label{q a+1ap 4 sum}
\end{align}
where $q_{0,a}(x)$, $[a q_{0,a-1}(x)]$ and $[a(a-1)(a-2) q_{0,a-3}(x)]$ are given by \eqref{def of q0ap1} and \eqref{q0 a negative}.
\end{lemma}

\begin{remark}\label{remark:degree of pol}
The degree of the polynomial in the right-hand side of \eqref{q a+1ap 4 sum} is given by 
\begin{equation*}
\begin{cases}
	2 & \textup{if }a=0,
	\\
	a & \textup{if }a\ge 1. 
\end{cases}
\end{equation*}
In particular, $\mathfrak{q}_{a+1,a,a+1}^{(4)}=\mathfrak{q}_{a+1,a,a+2}^{(4)}=0.$ 
\end{remark}

\begin{proof}
Let us first rewrite the sums on the left-hand sides of \eqref{q aap 4 sum} and \eqref{q a+1ap 4 sum} in terms of the coefficients $\{\mathfrak{q}_{\ell,m,p}^{(1)},\mathfrak{q}_{\ell,m,p}^{(2)},\mathfrak{q}_{\ell,m,p}^{(3)}\}$. By \eqref{q lmp 5 2 3}, 
\begin{align*}
& \sum_{m=0}^{a+1}\sum_{p=0}^{3a+2-2m} \mathfrak{q}_{a+1,m,p}^{(4)}k^{m}x^{p}  = - \sum_{s=0}^{a} \Bigg\{ \sum_{m=0}^{s}\sum_{p=0}^{s} \mathfrak{q}_{s,m,p}^{(1)}k^{m}x^{p} \nonumber \\
&  \times \bigg\{ \sum_{m=0}^{a+1-s}\sum_{p=0}^{ 3(a+1-s)-1-2m } \hspace{-0.3cm} \mathfrak{q}_{a+1-s,m,p}^{(2)}k^{m}x^{p} -\frac{1}{3r^{b}} \sum_{m=0}^{a-s}\sum_{p=0}^{3(a-s)-2m} \mathfrak{q}_{a-s,m,p}^{(3)}k^{m}x^{p} \bigg\} \Bigg\}. 
\end{align*}
Equaling the coefficients of $k^{a+1}$ and of $k^{a}$ gives the identities
\begin{align}\label{lol43}
& \sum_{p=0}^{a} \mathfrak{q}_{a+1,a+1,p}^{(4)}x^{p} = - \sum_{s=0}^{a} \Bigg\{ \sum_{p=0}^{s} \mathfrak{q}_{s,s,p}^{(1)}x^{p} \times \bigg\{ \sum_{p=0}^{ a-s } \mathfrak{q}_{a+1-s,a+1-s,p}^{(2)}x^{p} \bigg\} \Bigg\}
\end{align}
and
\begin{align}
& \sum_{p=0}^{a+2} \mathfrak{q}_{a+1,a,p}^{(4)}x^{p} \hspace{-0.05cm} = \hspace{-0.05cm}  - \hspace{-0.05cm} \sum_{s=1}^{a} \hspace{-0.05cm} \Bigg\{ \hspace{-0.05cm} \sum_{p=0}^{s} \mathfrak{q}_{s,s-1,p}^{(1)}x^{p}  \hspace{-0.05cm} \times \hspace{-0.05cm} \bigg\{ \hspace{-0.05cm} \sum_{p=0}^{ a-s }  \mathfrak{q}_{a+1-s,a+1-s,p}^{(2)}x^{p} \bigg\} \hspace{-0.05cm}\Bigg\} \nonumber \\
& - \sum_{s=0}^{a} \Bigg\{ \sum_{p=0}^{s} \mathfrak{q}_{s,s,p}^{(1)}x^{p} \times \bigg\{ \sum_{p=0}^{ a-s+2 } \mathfrak{q}_{a+1-s,a-s,p}^{(2)}x^{p} -\frac{1}{3r^{b}} \sum_{p=0}^{a-s} \mathfrak{q}_{a-s,a-s,p}^{(3)}x^{p} \bigg\} \Bigg\}. \label{lol44}
\end{align}
It remains to simplify the right-hand sides of \eqref{lol43} and \eqref{lol44}. Sums of the form $\sum_{p=0}^{s} \mathfrak{q}_{s,s,p}^{(1)}x^{p}$ and $\sum_{p=0}^{s} \mathfrak{q}_{s,s-1,p}^{(1)}x^{p} $ were already simplified in \eqref{q aap 1 sum} and \eqref{q a+1ap 1 sum}. Also, by \eqref{mathfrak qlmp(2)},
\begin{equation} \label{q ll 2}
\sum_{p=0}^{\ell-1} \mathfrak{q}_{\ell,\ell,p}^{(2)} \Big( \frac{x}{r^b}\Big)^{p} = \sum_{s=1}^{\ell} \frac{ 2^{ s/2 } }{s!} \He_{s-1} ( x ) \frac{1}{\ell!}\Big[\frac{d}{dk}\Big]^{\ell}\Big[[(\Xi_{j,k}^{\mathrm{formal}})^s]_{\ell}\Big]\Big|_{k=0,M_{j}\to \frac{x}{r^{b}}}
\end{equation}
where $[(\Xi_{j,k}^{\mathrm{formal}})^s]_{\ell}$ is the coefficient of the term of order $n^{-\frac{\ell}{2}}$ in the asymptotic series $(\Xi_{j,k}^{\mathrm{formal}})^s$. Using \eqref{Xi expansion}, we infer that
\begin{align}\label{lol42}
\frac{1}{\ell!}\Big[\frac{d}{dk}\Big]^{\ell}\Big[[(\Xi_{j,k}^{\mathrm{formal}})^s]_{\ell}\Big]\Big|_{k=0} = \delta_{s,\ell} d_{1,1}^{\ell}, \qquad s=1,\ldots,\ell.
\end{align}
Also, a direct computation using \eqref{def of Xijk} shows that
\begin{align}\label{Xi expansion lol}
\Xi_{j,k} = \frac{1}{2\sqrt{2} b r^{b} \sqrt{n}}\bigg\{ k + \frac{5 (r^{b}M_{j})^{2}}{3}b + \frac{1}{\sqrt{n}} \bigg( \frac{k M_{j}}{3} - \frac{53}{36}bM_{j}^{3}r^{2b} \bigg) + \bigO(n^{-1}) \bigg\}                     
\end{align}
as $n \to + \infty$ uniformly for $j \in \{g_{-},\ldots,g_{+}\}$. In particular, $d_{1,1} = \frac{1}{2\sqrt{2} b r^{b}}$ and thus, by \eqref{q ll 2}-\eqref{lol42}, 
\begin{equation} \label{q ll 2 new}
\sum_{p=0}^{\ell-1} \mathfrak{q}_{\ell,\ell,p}^{(2)} \Big( \frac{x}{r^b}\Big)^{p} = \frac{1}{ (2b r^b)^{\ell} }  \frac{1}{ \ell! } \He_{\ell-1}(x), \qquad \ell \geq 1.
\end{equation}
Combining \eqref{lol43} with \eqref{def of p0a}, \eqref{q aap 1 sum} and \eqref{q ll 2 new} yields 
\begin{align*}
(2br^{b})^{a+1} (-1)^{a+1} (a+1)! \sum_{p=0}^{a} \mathfrak{q}_{a+1,a+1,p}^{(4)}  \Big(\frac{x}{r^b}\Big)^{p} = (-1)^a \sum_{s=0}^a  \binom{a+1}{s}  i^s \He_s(ix)  \He_{a-s}(x).
\end{align*}
Using the functional equation (see e.g. \cite[eq.(9.6)]{Wunsche})
\begin{equation}\label{Heap1p formula 1}
\sum_{s=0}^{a}  \binom{a+1}{s}  i^s \He_s(ix)  \He_{a-s}(x) =  i^a\, \He_a^{(1)}(ix), \qquad a \in \mathbb{N},
\end{equation}
we obtain the desired identity \eqref{q aap 4 sum}. It remains to simplify the right-hand side of \eqref{lol44} (with $x$ replaced by $\frac{x}{r^{b}}$) and to prove \eqref{q a+1ap 4 sum}. In view of \eqref{q aap 1 sum}, \eqref{q a+1ap 1 sum} and \eqref{q ll 2 new}, it only remains to evaluate explicitly sums of the forms
\begin{align}\label{two remaining sums}
\sum_{p=0}^{\ell+2} \mathfrak{q}_{\ell+1,\ell,p}^{(2)} \Big( \frac{x}{r^b}\Big)^{p} \quad \mbox{ and } \quad \frac{1}{3 r^{b} }\sum_{p=0}^{\ell} \mathfrak{q}_{\ell,\ell,p}^{(3)} \Big( \frac{x}{r^b}\Big)^{p}, \qquad \ell \geq 0.
\end{align}
For the first sum, we use \eqref{mathfrak qlmp(2)} to get
\begin{align}\label{first sum first expr}
\sum_{p=0}^{\ell+2} \mathfrak{q}_{\ell+1,\ell,p}^{(2)} \Big( \frac{x}{r^b}\Big)^{p} = \sum_{s=1}^{\ell+1} \frac{ 2^{ s/2 } }{s!} \He_{s-1} ( x ) \frac{1}{\ell!}\Big[\frac{d}{dk}\Big]^{\ell}\Big[[(\Xi_{j,k}^{\mathrm{formal}})^s]_{\ell+1}\Big]\Big|_{k=0,M_{j}\to \frac{x}{r^{b}}}, \qquad \ell \geq 0.
\end{align}
A direct computation using \eqref{Xi expansion} shows that
\begin{align*}
\frac{1}{\ell!}\Big[\frac{d}{dk}\Big]^{\ell}\Big[[(\Xi_{j,k}^{\mathrm{formal}})^s]_{\ell+1}\Big]\Big|_{k=0} = \begin{cases}
0, & \mbox{if } 1 \leq s \leq \ell-2, \\
(\ell-1) d_{1,1}^{\ell-2}d_{3,2}, & \mbox{if } s = \ell-1, \\
\ell \frac{x}{r^{b}} d_{1,1}^{\ell-1}d_{2,1}, & \mbox{if } s = \ell, \\
(\ell+1) \frac{x^{2}}{r^{2b}} d_{1,1}^{\ell}d_{1,0}, & \mbox{if } s = \ell+1,
\end{cases}
\end{align*}
and by \eqref{def of Xijk} and \eqref{Xi expansion lol},
\begin{align*}
d_{1,1} = \frac{1}{2\sqrt{2} b r^{b}}, \qquad d_{1,0} = \frac{5r^{b}}{6\sqrt{2}}, \qquad d_{2,1} = \frac{1}{6\sqrt{2} br^{b}}, \qquad d_{3,2} = \frac{-1}{24\sqrt{2}b^{2}r^{3b}}.
\end{align*}
Substituting the above in \eqref{first sum first expr}, for $\ell \geq 0$ we get
\begin{align}
\sum_{p=0}^{\ell+2} \mathfrak{q}_{\ell+1,\ell,p}^{(2)} \Big( \frac{x}{r^b}\Big)^{p} & = \frac{1}{(2b)^\ell r^{b(\ell+1) }  }\frac{1}{ \ell! } \Big(  \frac56 \, x^2 \, \He_{\ell}(x) + \frac{\ell}{3} \,  x\, \He_{\ell-1}(x) -\frac{\ell(\ell-1)}{6} \He_{\ell-2}(x)  \Big), \nonumber \\
& = \frac{1}{(2b)^\ell r^{b(\ell+1) }  }\frac{1}{ \ell! } \Big[ \frac{5x^{2}+2\ell}{6} \, \He_{\ell}(x) + \frac{\ell(\ell-1)}{6} \He_{\ell-2}(x)  \Big], \label{sum2 lol}
\end{align}
where $\He_{-2}(x) \equiv \He_{-1}(x) \equiv 0$, and for the second line we have used the three-term recurrence relation \eqref{Hermite recurrence} of $\He_\ell$. 

Now we turn to the problem of simplifying the second sum in \eqref{two remaining sums}. Let us write
\begin{align*}
\sqrt{n}\bigg(\sum_{\ell=0}^{+\infty} \frac{c_{\ell}(\eta_{j,k})}{a_{j,k}^{\ell+\frac{1}{2}}} \bigg) \sim \sum_{\ell=0}^{+\infty} \frac{e_{\ell}}{n^{\frac{\ell}{2}}}, \qquad \mbox{as } n \to + \infty
\end{align*}
for some $\{e_{\ell}\}_{\ell=0}^{+\infty} \subset \mathbb{C}$. By \eqref{mathfrak qlmp(3)}, 
\begin{align}
& \frac{1}{3r^{b}}\sum_{p=0}^{\ell} \mathfrak{q}_{\ell,\ell,p}^{(3)} = - \frac{1}{\ell!}\Big[\frac{d}{dk}\Big]^{\ell} \bigg\{ e_{\ell} + \sum_{s=1}^{\ell} e_{\ell-s} \sum_{m=1}^{s} \frac{ 2^{ m/2 } }{m!} \He_{m} ( x ) \Big[[(\Xi_{j,k}^{\mathrm{formal}})^m]_{s}\Big] \bigg\}\Big|_{k=0,M_{j}\to \frac{x}{r^{b}}} \nonumber \\
& = - \frac{1}{\ell!} \bigg\{ \Big[\frac{d}{dk}\Big]^{\ell}e_{\ell} + \sum_{s=1}^{\ell}  \sum_{m=1}^{s} \frac{ 2^{ m/2 } }{m!} \He_{m} ( x ) \sum_{q=0}^{\ell}\binom{\ell}{q}\Big[\frac{d}{dk}\Big]^{q}e_{\ell-s} \; \Big[\frac{d}{dk}\Big]^{\ell-q}\Big[[(\Xi_{j,k}^{\mathrm{formal}})^m]_{s}\Big] \bigg\}\Big|_{k=0,M_{j}\to \frac{x}{r^{b}}}. \label{lol45}
\end{align}
Long but direct calculations using \eqref{Xi l expansion}, \eqref{ajk expansion} and \eqref{etajk expansion} show that
\begin{align*}
& \Big[\frac{d}{dk}\Big]^{q}e_{\ell-s}\Big|_{k=0} = 0, & & \mbox{for } \; \ell \geq 0, \; 0 \leq s \leq \ell, \; q \geq 1+\Big\lfloor \frac{\ell-s}{2} \Big\rfloor, \\
& \Big[\frac{d}{dk}\Big]^{\ell-q}\Big[[(\Xi_{j,k}^{\mathrm{formal}})^m]_{s}\Big]\Big|_{k=0} = 0, & & \mbox{for } \; \ell \geq 1, \; 1 \leq s \leq \ell, \; 1 \leq m \leq s, \; q \leq \Big\lfloor \ell-\frac{s+m+1}{2} \Big\rfloor.
\end{align*}
This means that for $\ell \geq 1$, the only term that contributes in \eqref{lol45} corresponds to $m=s=\ell$ and $q=0$. Thus, for any $\ell \geq 0$, we have
\begin{align*}
\frac{1}{3r^{b}}\sum_{p=0}^{\ell} \mathfrak{q}_{\ell,\ell,p}^{(3)} = -  \bigg\{    e_{0}\frac{ 2^{ \ell/2 } }{\ell!} \He_{\ell} ( x )  \; \frac{1}{\ell!}\Big[\frac{d}{dk}\Big]^{\ell}\Big[[(\Xi_{j,k}^{\mathrm{formal}})^\ell]_{\ell}\Big] \bigg\}\Big|_{k=0,M_{j}\to \frac{x}{r^{b}}}.
\end{align*}
A direct computation shows that $e_{0}=-\frac{1}{3r^{b}}$. Using also \eqref{lol42}, we get
\begin{equation}\label{sum3 lol}
\frac{1}{3r^{b}}\sum_{p=0}^{\ell} \mathfrak{q}_{\ell,\ell,p}^{(3)} \Big( \frac{x}{r^b}\Big) ^{p}= \frac13 \frac{1}{(2b)^{\ell} r^{b(\ell+1) }  }\frac{1}{ \ell! } \He_{\ell}(x), \qquad \ell \geq 0.
\end{equation}
Using \eqref{sum2 lol} and \eqref{sum3 lol}, for $\ell \geq 0$ we obtain 
\begin{align}
& \sum_{p=0}^{\ell+2} \mathfrak{q}_{\ell+1,\ell,p}^{(2)}  \Big( \frac{x}{r^b}\Big) ^{p} -\frac{1}{3r^{b}} \sum_{p=0}^{\ell} \mathfrak{q}_{\ell,\ell,p}^{(3)}  \Big( \frac{x}{r^b}\Big) ^{p} \nonumber \\
& =  \frac{1}{(2b)^{\ell} r^{b(\ell+1) }  }\frac{1}{\ell! }  \Big[ \frac{5x^{2}+2(\ell-1)}{6}  \, \He_{\ell}(x) + \frac{\ell(\ell-1)}{6} \He_{\ell-2}(x)  \Big]. \label{q l l-1 23}
\end{align}
Combining \eqref{lol44} with \eqref{q aap 1 sum}, \eqref{q a+1ap 1 sum} and \eqref{q l l-1 23}, we obtain after some calculations that
\begin{equation} \label{q hat sum p hat}
(2br^{b})^{a+1} (-1)^a a!  \sum_{p=0}^{a+2} \mathfrak{q}_{a+1,a,p}^{(4)}  \Big( \frac{x}{r^b}\Big)^{p} =- \frac{b}{3}\bigg\{ \sum_{s=1}^a (-1)^{a-s} \binom{a}{s-1} \widehat{p}_{0,s}(x)  \He_{a-s}(x) +\mathfrak{s}_a(x) \bigg\}, 
\end{equation} 
where 
\begin{align}
& \widehat{p}_{0,a+1}(x):=  3 a\,p_{0,a+1}(x)+ 3a(1-3a) \,p_{0,a-1}(x) + 5a (a-1)(a-2) \, p_{0,a-3}(x), \label{def of p0hat}\\
& \mathfrak{s}_a(x) := (-1)^{a}\sum_{s=0}^a  \binom{a}{s} i^s\He_{s}(ix) \Big[ ( 5 x^2+2(a-s-1))  \, \He_{a-s}(x) + (a-s)(a-s-1) \He_{a-s-2}(x)  \Big] .
\end{align}
The polynomial $\mathfrak{s}_a$ can actually be considerably simplified. Indeed, since 
\begin{equation*}
	\He_s(x)= \Big[ \frac{d}{dt} \Big]^s \Big[ e^{xt-\frac{t^2}{2}}\Big] \Big|_{t=0}, \qquad  i^s \He_s(ix)=\Big[ \frac{d}{dt} \Big]^s \Big[ e^{-xt+\frac{t^2}{2}}\Big] \Big|_{t=0},
\end{equation*} 
we have  
\begin{equation} \label{Hermite sum vanishing}
\sum_{s=0}^a  \binom{a}{s} i^s\He_s(ix) \He_{a-s}(x)= \Big[ \frac{d}{dt} \Big]^a \Big[ 1 \Big] \Big|_{t=0}= \begin{cases}
1 & \textup{if }a=0, \\
0 & \textup{if }a \ge 1.
\end{cases} 
\end{equation}
Hence,
\begin{multline}\label{Hermite sum2}
\sum_{s=0}^a \binom{a}{s} (a-s)(a-s-1) i^s\He_{s}(ix)  \He_{a-s-2}(x) \\
=a(a-1) \sum_{s=0}^{a-2} \binom{a-2}{s} i^s\He_{s}(ix)  \He_{a-s-2}(x)
= \begin{cases}
2 & \textup{if } a=2, \\
0 &\textup{if } a\not=2,
\end{cases} 
\end{multline}
and using also the recurrence relation \eqref{Hermite recurrence} we get
\begin{multline}\label{Hermite sum3}
\sum_{s=0}^a \binom{a}{s} (a-s) i^s\He_{s}(ix)  \He_{a-s}(x) \\
=a \sum_{s=0}^{a-1}  i^s \binom{a-1}{s}  \He_{s}(ix) \Big( x \He_{a-s-1}(x)-(a-s-1)\He_{a-s-2}(x) \Big)
= \begin{cases}
x & \textup{if } a=1, \\
-2 &\textup{if }a=2, \\
0 &\textup{otherwise}.
\end{cases} 
\end{multline}
Using \eqref{Hermite sum vanishing}, \eqref{Hermite sum2} and \eqref{Hermite sum3} to simplify $\mathfrak{s}_{a}$, we finally obtain
\begin{equation} \label{mathfrak s evaluation}
\mathfrak{s}_a(x)= \begin{cases}
5x^2-2 & \textup{if } a=0, \\
-2x & \textup{if } a=1, \\
-2 & \textup{if } a=2, \\
0  & \textup{if } a\geq 3. 
\end{cases}
\end{equation}
Let us now simplify the sum in \eqref{q hat sum p hat}. First, substituting the definitions \eqref{def of p0hat} and \eqref{def of p0a}, we rewrite it as
\begin{align}\label{lol47}
& \sum_{s=1}^a (-1)^{a-s} \binom{a}{s-1} \widehat{p}_{0,s}(x) \He_{a-s}(x) = A_{1} + A_{2} + A_{3},
\end{align}
where
\begin{align*}
& A_{1} := 3(-1)^{a} \sum_{s=2}^a \binom{a}{s-1} (s-1) i^s \He_{s}(ix)\He_{a-s}(x), \\
& A_{2} := -3(-1)^a \sum_{s=2}^a \binom{a}{s-1} (s-1)\Big( 3(s-2)+2\Big) i^{s-2} \He_{s-2}(ix)\He_{a-s}(x), \\
& A_{3} := 5(-1)^a \sum_{s=4}^a \binom{a}{s-1} (s-1)(s-2)(s-3) i^{s} \He_{s-4}(ix)\He_{a-s}(x).
\end{align*}
To simplify $A_{1}$, we first establish two formulas. Using \eqref{Hermite sum vanishing} and $\binom{a+1}{s}=\binom{a}{s}+\binom{a}{s-1}$ in \eqref{Heap1p formula 1}, we infer that
\begin{align}\label{Heap1p formula 2}
\sum_{s=1}^a  \binom{a}{s-1} i^s\He_s(ix)  \He_{a-s}(x) = \begin{cases}
i^a\, \He_a^{(1)}(ix), & \mbox{if } a \geq 1, \\
0, & \mbox{if } a=0.
\end{cases}
\end{align}
Also, using the recurrence \eqref{Hermite recurrence} together with \eqref{Hermite sum vanishing}, 
\begin{align}
a \sum_{s=1}^{a} \binom{a-1}{s-1} i^{s} \He_{s}(ix) \He_{a-s}(x) & = a\sum_{s=1}^{a} \binom{a-1}{s-1} i^{s} \Big( i x \He_{s-1}(ix)-(s-1)\He_{s-2}(ix) \Big) \He_{a-s}(x) \nonumber \\
& = \begin{cases}
-x, & \mbox{if } a = 1, \\
2, & \mbox{if } a=2, \\
0, & \mbox{otherwise.}
\end{cases}\label{lol46}
\end{align}
For $A_{1}$, we use $\binom{a-1}{s-2}=\binom{a}{s-1}-\binom{a-1}{s-1}$ together with \eqref{Heap1p formula 2} and \eqref{lol46}, and find
\begin{align}\label{A1 simplif}
A_{1}= 3(-1)^{a} a \sum_{s=2}^a \binom{a-1}{s-2}  i^s \He_{s}(ix)\He_{a-s}(x) = 3 a \,q_{0,a+1}(x) + \begin{cases}
-3x & \textup{if } a=1, \\
-6 & \textup{if } a=2, \\
0 & \textup{otherwise}. 
\end{cases}
\end{align}
Similarly, by \eqref{Heap1p formula 1} and \eqref{Heap1p formula 2},
\begin{align}
A_{2} &= 3(-1)^a \sum_{s=2}^a \Big[3a(a-1)\binom{a-2}{s-3} + 2a\binom{a-1}{s-2}  \Big] i^{s-2} \He_{s-2}(ix)\He_{a-s}(x) \nonumber \\
&= 3(1-3a) \,[a q_{0,a-1}(x)]+ \begin{cases}
-3 &\textup{if } a=0, \\
18 &\textup{if } a=2, \\
0 & \textup{otherwise}
\end{cases} 	\label{A2 simplif}
\end{align}
Simplifying $A_{3}$ is a simpler task as it only relies on \eqref{Heap1p formula 1}, namely
\begin{align}
A_{3} & = 5(-1)^a a(a-1)(a-2)\sum_{s=0}^{a-4} \binom{a-3}{s} i^{s} \He_{s}(ix)\He_{a-4-s}(x) = \begin{cases}
0, & \hspace{-1cm} \mbox{if } a\in \{0,1,2\} \\
5a(a-1)(a-2)q_{0,a-3}(x), & \hspace{-0.1cm} \mbox{if } a \geq 3
\end{cases} \nonumber \\
& = 5[a(a-1)(a-2)q_{0,a-3}(x)] + \begin{cases}
5(1-x^{2}), & \mbox{if } a=0, \\
5x, & \mbox{if } a = 1, \\
-10, & \mbox{if } a = 2, \\
0, & \mbox{otherwise.}
\end{cases} \label{A3 simplif}
\end{align}
Therefore, by \eqref{lol47}, \eqref{A1 simplif}, \eqref{A2 simplif} and \eqref{A3 simplif}, we have 
\begin{align}
& \sum_{s=1}^a (-1)^{a-s} \binom{a}{s-1} \widehat{p}_{0,s}(x) \He_{a-s}(x) \nonumber \\
&= 3 a \,q_{0,a+1}(x)+3 (1-3a) [a\,q_{0,a-1}(x)]+5[a(a-1)(a-2)\,q_{0,a-3}(x)]-\mathfrak{s}_a(x). \label{q hat sum p hat v2}
\end{align} 
Now \eqref{q a+1ap 4 sum} directly follows from \eqref{q hat sum p hat} and \eqref{q hat sum p hat v2}, which completes the proof.
\end{proof}
In the large $n$ asymptotics of $S_{2}^{(2)}$, which are obtained in Lemma \ref{lemma:S2kp2p} below, the function $\mathcal{G}_{0}(y; u,a)$, defined in \eqref{function F}, will appear inside a logarithm and in a denominator. The next lemma ensures that this function is positive for relevant values of the parameters.
\begin{lemma} \label{lemma:mathcalG0 positive}
The function $\mathcal{G}_{0}(y; u,a)$ given in \eqref{function F} is positive for all $y \in \R$, $u \in \R$, and $a \in \mathbb{N}$. 
\end{lemma}
\begin{proof}
Let us write 
\begin{align*}
\mathcal{G}_{0,0}(y,a)&:= (-1)^{a}  \bigg[ p_{0,a}(-\sqrt{2}y)\frac{2-\mathrm{erfc}(y)}{2}  - q_{0,a}(-\sqrt{2} y)  \frac{e^{-y^{2}}}{\sqrt{2\pi}} \bigg],
\\
\mathcal{G}_{0,1}(y,a)&:=  p_{0,a}(-\sqrt{2}y) \mathrm{erfc}(y)  + 2\,q_{0,a}(-\sqrt{2} y)  \frac{e^{-y^{2}}}{\sqrt{2\pi}} .
\end{align*}
Then by \eqref{function F}, we have 
\begin{equation} \label{mathcal G0 decomposition}
\mathcal{G}_{0}(y; u,a)=\mathcal{G}_{0,0}(y,a)+\frac{e^{u}}{2}\mathcal{G}_{0,1}(y,a).
\end{equation}
Using the definitions \eqref{def of p0a}, \eqref{def of q0ap1}, we obtain
\begin{equation} \label{p0 q0 differentiation rule}
	p_{0,a+1}'(x)=(a+1) p_{0,a}(x), \qquad 
	q_{0,a+1}'(x)=(a+1) q_{0,a}(x)+x q_{0,a+1}(x)-p_{0,a+1}(x). 
\end{equation}
The first identity in \eqref{p0 q0 differentiation rule} is well-known and easy to prove. The second one follows from two known identities, namely the differentiation rule $(\He_{k}^{(\nu)})'(x) = (k+\nu)\He_{k-1}^{(\nu)}(x)-\nu \He_{k-1}^{(\nu+1)}(x)$ and the recurrence relation $\He_{k+1}^{(\nu-1)}(x) = x \He_{k}^{(\nu)}(x)-\nu \He_{k-1}^{(\nu+1)}(x)$, see e.g. \cite[eqs (6.3) and (6.6)]{Wunsche}.
It follows from \eqref{p0 q0 differentiation rule} that
\begin{equation} \label{G01 derivative}
\frac{d}{dy} 	\mathcal{G}_{0,0}(y,a+1) =  \sqrt{2}(a+1) \mathcal{G}_{0,0}(y,a) ,\qquad 	\frac{d}{dy} 	\mathcal{G}_{0,1}(y,a+1) = - \sqrt{2}(a+1) \mathcal{G}_{0,1}(y,a).
\end{equation}
Let us now show that 
\begin{align} 
& \mathcal{G}_{0,1}(y,a) >0, & & y\in \R, \label{G01 positive} \\
& \mathcal{G}_{0,1}(y,a) = \sqrt{ \frac{2}{\pi} } e^{-y^2} \frac{a!}{ (\sqrt{2}y)^{a+1} } \Big(1+O(y^{-2}) \Big), & & \mbox{as } y \to +\infty. \label{asymp at inf G01}
\end{align}
We shall prove \eqref{G01 positive} and \eqref{asymp at inf G01} by induction on $a$. For $a=0$, we have $\mathcal{G}_{0,1}(y,0)=\mathrm{erfc}(y)$ and \eqref{G01 positive}, \eqref{asymp at inf G01} follow. Assume now that \eqref{G01 positive} and \eqref{asymp at inf G01} hold for a given $a$. By combining \eqref{G01 derivative} with \eqref{asymp at inf G01}, we infer that \eqref{asymp at inf G01} also holds with $a$ replaced by $a+1$; in particular $\mathcal{G}_{0,1}(y,a+1)>0$ for all sufficiently large $y>0$. Also, by \eqref{G01 derivative} and \eqref{G01 positive}, the function $\mathcal{G}_{0,1}(y,a+1)$ is decreasing for $y \in \R$. Since $\mathcal{G}_{0,1}(y,a+1)>0$ for all sufficiently large $y>0$, we conclude that $\mathcal{G}_{0,1}(y,a+1) >0$ for all $y \in \mathbb{R}$.  

The proof that $\mathcal{G}_{0,0}(y,a) >0$ for all $y\in \mathbb{R}$ is similar, so we omit it.

The statement of the lemma now readily follows from \eqref{mathcal G0 decomposition}.
\end{proof}

To obtain the large $n$ asymptotics of $\smash{S_{2}^{(2)}}$, we need the following lemma from \cite{Charlier 2d gap} to approximate sums of the form $\sum_{j}h(M_{j})$ (by contrast, Lemma \ref{lemma:Riemann sum NEW} approximates sums of the form $\sum_{j}f(j/n)$).
\begin{lemma}(\cite[Lemma 3.11]{Charlier 2d gap})\label{lemma:Riemann sum}
Let $h \in C^{3}(\mathbb{R})$ and $k \in \{1,\ldots,2g\}$. As $n \to + \infty$, we have
\begin{align}
& \sum_{j=g_{-}}^{g_{+}}h(M_{j}) = br^{2b} \int_{-M}^{M} h(t) dt \; \sqrt{n} - 2 b r^{2b} \int_{-M}^{M} th(t) dt + \bigg( \frac{1}{2}-\theta_{-}^{(n,M)} \bigg)h(M)+ \bigg( \frac{1}{2}-\theta_{+}^{(n,M)} \bigg)h(-M) \nonumber \\
& + \frac{1}{\sqrt{n}}\bigg[ 3br^{2b} \int_{-M}^{M}t^{2}h(t)dt + \bigg( \frac{1}{12}+\frac{\theta_{-}^{(n,M)}(\theta_{-}^{(n,M)}-1)}{2} \bigg)\frac{h'(M)}{br^{2b}} - \bigg( \frac{1}{12}+\frac{\theta_{+}^{(n,M)}(\theta_{+}^{(n,M)}-1)}{2} \bigg)\frac{h'(-M)}{br^{2b}} \bigg] \nonumber \\
& + \bigO\Bigg(  \frac{1}{n^{3/2}} \sum_{j=g_{-}+1}^{g_{+}} \bigg( (1+|M_{j}|^{3}) \tilde{\mathfrak{m}}_{j,n}(h) + (1+M_{j}^{2})\tilde{\mathfrak{m}}_{j,n}(h') + (1+|M_{j}|) \tilde{\mathfrak{m}}_{j,n}(h'') + \tilde{\mathfrak{m}}_{j,n}(h''') \bigg)   \Bigg), \label{sum f asymp 2}
\end{align}
where, for $\tilde{h} \in C(\mathbb{R})$ and $j \in \{g_{-}+1,\ldots,g_{+}\}$, $\tilde{\mathfrak{m}}_{j,n}(\tilde{h}) := \max_{x \in [M_{j},M_{j-1}]}|\tilde{h}(x)|$.
\end{lemma}

\begin{lemma}\label{lemma:S2kp2p}
As $n \to + \infty$, we have
\begin{align*}
&  S_{2}^{(2)} = -abr^{2b}M\sqrt{n} \ln n + br^{2b} \sqrt{n} \int_{-M}^{M}h_{0}(t)dt + a\frac{\theta_{-}^{(n,M)}+\theta_{+}^{(n,M)}-1}{2}\ln n \\
& + br^{2b} \int_{-M}^{M} \big( h_{1}(t)-2th_{0}(t) \big)dt + \bigg( \frac{1}{2}-\theta_{-}^{(n,M)} \bigg)h_{0}(M) + \bigg( \frac{1}{2}-\theta_{+}^{(n,M)} \bigg)h_{0}(-M) + \bigO\bigg( \frac{M^{3}}{\sqrt{n}}\ln n \bigg),
\end{align*}
where
\begin{align}
	\mathcal{H}_{0}(x) &:=\frac{r^{a-ab}}{(2b)^{a}} \, \mathcal{G}_0\Big( \hspace{-0.1cm}-\frac{r^bx}{\sqrt{2}};u,a\Big),  \label{def of Hmathcal0}
	\\
	 h_{0}(x) &:= \ln \big(\mathcal{H}_{0}(x)\big), \label{def of h0} 
	\\
	  h_1(x)&: =  \frac{r^{a-(1+a)b}}{(2b)^{a+1}} \,\frac{1}{\mathcal{H}_0(x)}\, \mathcal{G}_1\Big( \hspace{-0.1cm}-\frac{r^bx}{\sqrt{2}};u,a\Big),  \label{def of h1}
\end{align}
and the functions $\mathcal{G}_0$ and  $\mathcal{G}_1$ are given by \eqref{function F} and \eqref{function H}. 
\end{lemma}

\begin{proof}
By \eqref{asymp prelim of S2kpvp} and Lemma \ref{lemma: uniform}, 
\begin{align}
S_{2}^{(2)} & = \sum_{j:\lambda_{j}\in I_{2}} \ln \bigg(\sum_{k=0}^{a} {a \choose k} \frac{(-r)^{a-k}}{n^{\frac{k}{2b}}} \frac{\Gamma(\frac{2j+2\alpha+k}{2b})}{\Gamma(\frac{2j+2\alpha}{2b})} \nonumber \\
& \hspace{3cm} \times \bigg[ 1+((-1)^{a}e^{u}-1)\bigg( \mathrm{erfc}\Big(\hspace{-0.1cm}-\hspace{-0.05cm}\eta_{j,k} \sqrt{\tfrac{a_{j,k}}{2}}\Big)  - R_{a_{j,k}}(\eta_{j,k}) \bigg) \bigg]\bigg). \label{lol30}
\end{align}
Recall that for all $j \in \{j:\lambda_{j}\in I_{2}\}$, we have $1-\frac{M}{\sqrt{n}} \leq \lambda_{j} = \frac{bnr^{2b}}{j+\alpha} \leq 1+\frac{M}{\sqrt{n}}$, and $-M \leq M_{j} \leq M$. 


The expansion \eqref{lol29} implies that
\begin{align}\label{expansion in terms of mathcalB}
& \sum_{k=0}^{a} {a \choose k} \frac{(-r)^{a-k}}{n^{\frac{k}{2b}}} \frac{\Gamma(\frac{2j+2\alpha+k}{2b})}{\Gamma(\frac{2j+2\alpha}{2b})} \bigg[ 1+((-1)^{a}e^{u}-1)\bigg( \frac{1}{2}\mathrm{erfc}\Big(\hspace{-0.1cm}-\hspace{-0.05cm}\eta_{j,k} \sqrt{\tfrac{a_{j,k}}{2}}\Big) \hspace{-0.1cm}-\hspace{-0.05cm} R_{a_{j,k}}(\eta_{j,k}) \bigg) \bigg] \sim \sum_{\ell=0}^{+\infty} \frac{\mathcal{B}_{\ell}(M_{j};a)}{\sqrt{n^{\ell} }},
\end{align}
where
\begin{align}\label{def of mathcalB}
\mathcal{B}_{\ell}(x;a) := \sum_{k=0}^{a} {a \choose k} (-1)^{a-k}r^{a}\mathcal{A}_{\ell}(x;k).
\end{align}
It turns out that the first $a$ terms in the expansion \eqref{expansion in terms of mathcalB} are 0, i.e.
\begin{align}\label{first mathcalB are 0}
\mathcal{B}_{0}(\cdot;a) \equiv \mathcal{B}_{1}(\cdot;a) \equiv \ldots \equiv \mathcal{B}_{a-1}(\cdot;a) \equiv 0.
\end{align}
To prove this, we use \cite[eq 26.8.6]{NIST}, i.e.
\begin{equation} \label{Stirling number sum}
	\sum_{k=0}^{a}{a \choose k}(-1)^{a-k}k^{\ell} = a! \,S(\ell,a), \qquad \ell,a \in \mathbb{N},
\end{equation} 
where we recall that $S(\ell,a)$ is the Stirling number of the second kind. In particular, 
\begin{equation}\label{sum tricky 0}
	\sum_{k=0}^{a}{a \choose k}(-1)^{a-k}k^{\ell} =  \begin{cases}
		0 &\mbox{if } \ell < a,
		\\
		a! &\mbox{if }\ell=a,
		\\
		(a+1)!\,a/2 &\mbox{if }\ell=a+1.
	\end{cases}
\end{equation}
Since $k\mapsto \mathcal{A}_{\ell}(x;k)$ is a polynomial of degree $\ell$ by \eqref{expansion of mathcalA}, the identities \eqref{first mathcalB are 0} directly follow from \eqref{def of mathcalB}. 

Let us now compute $\mathcal{B}_a(x;a)$ and $\mathcal{B}_{a+1}(x;a)$. By \eqref{expansion of mathcalA}, \eqref{def of mathcalB} and \eqref{sum tricky 0}, we have 
\begin{align}
\mathcal{B}_a(x;a)& =r^a  \bigg( (-1)^a+\frac{e^{u}-(-1)^a}{2}\mathrm{erfc}\bigg( -\frac{r^{b}x}{\sqrt{2}} \bigg) \bigg) (-1)^a a! \sum_{p=0}^{a}\mathfrak{q}_{a,a,p}^{(1)} x^{p} \nonumber
\\
&+ r^a (e^{u}-(-1)^{a})\frac{e^{-\frac{r^{2b}x^{2}}{2}}}{\sqrt{2\pi} }(-1)^{a} a! \sum_{p=0}^{a-1} \mathfrak{q}_{a,a,p}^{(4)} x^{p} \label{Ba qaap (4,5)}
\end{align}
and 
\begin{align}
\mathcal{B}_{a+1}(x;a)& = r^a  \bigg( (-1)^a+\frac{e^{u}-(-1)^a}{2}\mathrm{erfc}\bigg( -\frac{r^{b}x}{\sqrt{2}} \bigg) \bigg) (-1)^a a!  \sum_{p=0}^{a+1} \Big( \mathfrak{q}_{a+1,a,p}^{(1)}+\frac{a(a+1)}{2} \mathfrak{q}_{a+1,a+1,p}^{(1)}  \Big) x^{p} \nonumber
		\\
		&+ r^a (e^{u}-(-1)^{a})\frac{e^{-\frac{r^{2b}x^{2}}{2}}}{\sqrt{2\pi} }(-1)^{a} a!  \bigg( \sum_{p=0}^{a+2} \mathfrak{q}_{a+1,a,p}^{(4)} x^{p} +  \frac{a(a+1)}{2}  \sum_{p=0}^{a}\mathfrak{q}_{a+1,a+1,p}^{(4)} x^{p} \bigg)  .  \label{Ba+1 qaap (4,5)}
\end{align}
More generally, for $\ell \geq 1$, we have
\begin{align}
\mathcal{B}_{a+\ell}(x;a) & = r^{a}\bigg( (-1)^{a}+\frac{e^{u}-(-1)^{a}}{2}\mathrm{erfc}\bigg( -\frac{r^{b}x}{\sqrt{2}} \bigg) \bigg)(-1)^{a}a! \sum_{m=\max(a,1)}^{a+\ell}\sum_{p=0}^{a+\ell}\mathfrak{q}_{a+\ell,m,p}^{(1)}S(m,a)x^{p} \nonumber \\
&\quad  + r^{a}(e^{u}-(-1)^{a})\frac{e^{-\frac{r^{2b}x^{2}}{2}}}{\sqrt{2\pi} }(-1)^{a}a! \sum_{m=a}^{a+\ell}\sum_{p=0}^{ 3(a+\ell)-1-2m } \mathfrak{q}_{a+\ell,m,p}^{(4)}S(m,a)x^{p}. \label{Ba+ell}
\end{align}
Let us define 
\begin{equation}\label{def of H0 h0 h1 Ba}
\mathcal{H}_{0}(x):=\mathcal{B}_{a}(x;a), \qquad h_{0}(x):=\ln(\mathcal{H}_{0}(x)), \qquad h_{1}(x):=\frac{\mathcal{B}_{a+1}(x;a)}{\mathcal{B}_{a}(x;a)}. 
\end{equation}
By combining \eqref{Ba qaap (4,5)} and \eqref{Ba+1 qaap (4,5)} with Lemmas~\ref{lemma:sum of qaa(1)} and \ref{lemma:sum of qaa(4)}, we obtain after simplifications that the functions $\mathcal{H}_{0}(x)$ and $h_{1}(x)$ can be written as in \eqref{def of Hmathcal0} and \eqref{def of h1}. 

Since the case $a=0$ was already done in \cite{Charlier}, from here on we focus on the more complicated case $a\ge 1$. We note the following important facts:
\begin{enumerate}[label=(\roman*)]
	\item By Lemma~\ref{lemma:mathcalG0 positive}, $\mathcal{H}_{0}(x)>0$ for all $x \in \mathbb{R}$.
	\item By \eqref{def of Hmathcal0} and \eqref{def of h0}, $h_{0}(x)$ grows logarithmically at $\pm\infty$.
	\item By \eqref{def of Hmathcal0} and \eqref{def of h1}, $h_{1}(x)$ grows linearly at $\pm \infty$.
	\item By \eqref{Ba+ell} and \eqref{Ba qaap (4,5)}, $\frac{\mathcal{B}_{a+\ell}(x;a)}{\mathcal{B}_{a}(x;a)}=\bigO(x^{\ell})$ as $x \to \pm \infty$.
\end{enumerate}
(Another reason as to why the case $a=0$ is significantly simpler than the case $a \geq 1$ stems from the fact that for $a=0$, the function $h_{0}(x)$ remains bounded, and furthermore $h_{1}(x)$ becomes exponentially small as $x\to \pm\infty$.)

After substituting \eqref{expansion in terms of mathcalB} in \eqref{lol30}, we obtain
\begin{align}\label{lol31}
S_{2}^{(2)} = -\frac{a}{2}\ln n \; \sum_{j=g_{-}}^{g_{+}}1 + \Sigma_{0}^{(2)} + \frac{1}{\sqrt{n}}\Sigma_{1}^{(2)} + \bigO \bigg( \frac{1}{n}\sum_{j=g_{k,-}}^{g_{k,+}}M_{j}^{2} \bigg),
\end{align}
where 
\begin{align*}
& \Sigma_{0}^{(2)} := \sum_{j=g_{-}}^{g_{+}} h_{0}(M_{j}), \qquad \Sigma_{1}^{(2)} := \sum_{j=g_{-}}^{g_{+}} h_{1}(M_{j}).
\end{align*}
The $\bigO$-term in \eqref{lol31} is $\bigO(\frac{M^{3}}{\sqrt{n}})$. Also, by Lemma \ref{lemma:Riemann sum}, 
\begin{align*}
& -\frac{a}{2}\ln n \; \sum_{j=g_{-}}^{g_{+}}1 = - a b r^{2b}M\sqrt{n} \ln n + a\frac{\theta_{-}^{(n,M)}+\theta_{+}^{(n,M)}-1}{2} \ln n + \bigO\bigg( \frac{M^{3}}{\sqrt{n}}\ln n \bigg), \\
& \Sigma_{0}^{(2)} = br^{2b}\sqrt{n}\int_{-M}^{M}h_{0}(t) \, dt -2br^{2b}\int_{-M}^{M}t \, h_{0}(t) \, dt \\
& \hspace{1cm} + \bigg( \frac{1}{2}-\theta_{-}^{(n,M)} \bigg)h_{0}(M) + \bigg( \frac{1}{2}-\theta_{+}^{(n,M)} \bigg)h_{0}(-M) + \bigO\bigg( \frac{M^{3}\ln n}{\sqrt{n}} \bigg), \\
& \frac{1}{\sqrt{n}}\Sigma_{1}^{(2)} = br^{2b} \int_{-M}^{M} h_{1}(t)dt + \bigO \bigg( \frac{M^{3}}{\sqrt{n}} \bigg), \qquad  \frac{1}{n}\sum_{j=g_{k,-}}^{g_{k,+}}M_{j}^{2} = \bigO\bigg( \frac{M^{3}}{\sqrt{n}} \bigg),
\end{align*}
as $n \to + \infty$. We now obtain the claim after a computation.
\end{proof}

\subsection{Asymptotics of $S_{2}$}
We are now in a position to compute the large $n$ asymptotics of $S_{2}$.
\begin{lemma}\label{lemma: asymp of S2k final}
As $n \to + \infty$,
\begin{align*}
& S_{2} =  \widehat{C}_{1}^{(\epsilon)} n + \widehat{C}_{2} \sqrt{n} + \widehat{C}_{3}^{(n,\epsilon)} + \bigO \bigg( \frac{M^{3}}{\sqrt{n}}\ln n + \frac{\sqrt{n}}{M^{5}} \bigg),
\end{align*}
where
\begin{align}
& \widehat{C}_{1}^{(\epsilon)} =  \int_{\frac{br^{2b}}{1+\epsilon}}^{br^{2b}} \bigg( u + a \ln \bigg( r-\bigg( \frac{x}{b} \bigg)^{\frac{1}{2b}} \bigg) \bigg)dx + \int_{br^{2b}}^{\frac{br^{2b}}{1-\epsilon}} a \ln \bigg( \bigg( \frac{x}{b} \bigg)^{\frac{1}{2b}} - r \bigg)dx, \label{def of C1 hat}
\\
& \widehat{C}_{2} = b r^{2b} \int_{-\infty}^{+\infty} \hat{h}_{0}(t) \, dt, \label{def of C2 hat}
\\
& \widehat{C}_{3}^{(n,\epsilon)} = -\alpha u + \frac{1+2\alpha - 2\theta_{-}^{(n,\epsilon)}}{2}\bigg( u + a \ln \Big( r-r(1+\epsilon)^{-\frac{1}{2b}} \Big) \bigg) \nonumber 
\\
& + \frac{1-2\alpha-2\theta_{+}^{(n,\epsilon)}}{2}a \ln \Big( r(1-\epsilon)^{-\frac{1}{2b}}-r \Big)   \nonumber
\\
&+ \frac{a}{4}\bigg\{ \frac{1-a}{1-(1+\epsilon)^{-\frac{1}{2b}}} + \frac{1-a}{(1-\epsilon)^{-\frac{1}{2b}}-1}  + (a-2b+4\alpha) \ln \bigg( \frac{(1-\epsilon)^{-\frac{1}{2b}}-1}{1-(1+\epsilon)^{-\frac{1}{2b}}} \bigg) \bigg\}  \nonumber
\\
&+ br^{2b} \int_{-\infty}^{+\infty} \bigg\{ h_{1}(t)-2\,th_{0}(t) +  \frac{a}{4b} \bigg( 1+2b+8b\ln \bigg( \frac{r|t|}{2b} \bigg) \bigg) t + 2u t \chi_{(0,\infty)} (t) - \frac{(2ab-a^{2})t}{4br^{2b}(1+t^{2})}  \bigg\} \,dt .  \nonumber 
\end{align}
Here
\begin{align*}
\hat{h}_{0}(x) = h_{0}(x) - a \ln \bigg( \frac{r|x|}{2b} \bigg) -u \,\chi_{(0,\infty)}(x)
\end{align*}
and $h_{0}, h_{1}$ are given in the statement of Lemma \ref{lemma:S2kp2p}.
\end{lemma}
\begin{proof}
Combining Lemmas \ref{lemma:S2kp1p}, \ref{lemma:S2kp3p} and \ref{lemma:S2kp2p} yields
\begin{align}\label{lol32}
& S_{2} = \widehat{C}_{1}^{(\epsilon)} + \widehat{C}_{2}'^{(M)} \sqrt{n} \ln n + \widehat{C}_{2}^{(M)} \sqrt{n} + \widehat{C}_{3}'^{(n,M)} \ln n + \widehat{C}_{3}^{(n,\epsilon,M)} + \bigO \bigg( \frac{M^{3}}{\sqrt{n}}\ln n + \frac{\sqrt{n}}{M^{5}} \bigg),
\end{align}
as $n \to +\infty$, where $\widehat{C}_{1}^{(\epsilon)}$ is given by \eqref{def of C1 hat}. After short (but remarkable) simplifications, we obtain $\widehat{C}'^{(M)}_{2} = \widehat{C}_{3}'^{(\epsilon,M)} = 0$. The quantity $\widehat{C}_{2}^{(M)}$ is given by
\begin{align*}
\widehat{C}_{2}^{(M)} = br^{2b}M\bigg( 2a \ln \bigg( \frac{2b}{r M} \bigg) + 2a -u \bigg) + \frac{a(a-1)b}{M} - \frac{ab (a-1)(2a-3) }{6r^{2b}M^{3}} + br^{2b}\int_{-M}^{M}h_{0}(t)dt.
\end{align*}
Using the definition \eqref{def of h0} of $h_{0}$, we verify that
\begin{align}\label{asymp of h0}
h_{0}(t) =  a \ln \Big( \frac{r|t|}{2b} \Big)+ u\chi_{(0,\infty)}(t) + \frac{a(a-1)}{2r^{2b}t^{2}} - \frac{a(a-1)(2a-3)}{4r^{4b}t^{4}} + \bigO(t^{-6}), \quad \mbox{as } t \to \pm \infty,
\end{align}
from which we conclude
\begin{align*}
\widehat{C}_{2}^{(M)} \sqrt{n} = \widehat{C}_{2} \sqrt{n} + \bigO\bigg( \frac{\sqrt{n}}{M^{5}} \bigg), \qquad \mbox{as } n \to + \infty,
\end{align*}
where $\widehat{C}_{2}$ is given by \eqref{def of C2 hat}. The quantity $\widehat{C}_{3}^{(n,\epsilon,M)}$ in \eqref{lol32} is given by
\begin{align*}
\widehat{C}_{3}^{(n,\epsilon,M)} = \bigg( \frac{1}{2}-\theta_{-}^{(n,M)} \bigg) \widehat{C}_{3,1}^{(M)} + \bigg( \frac{1}{2}-\theta_{+}^{(n,M)} \bigg) \widehat{C}_{3,2}^{(M)} + \widehat{C}_{3,3}^{(n,\epsilon)} + \widehat{C}_{3,4}^{(M)},  
\end{align*}
where
\begin{align*}
& \widehat{C}_{3,1}^{(M)} := h_{0}(M) - u - a \ln \bigg( \frac{rM}{2b} \bigg), \qquad \widehat{C}_{3,2}^{(M)} := h_{0}(-M)  - a \ln \bigg( \frac{rM}{2b} \bigg), \\
& \widehat{C}_{3,3}^{(n,\epsilon)} := \frac{1+2\alpha - 2\theta_{-}^{(n,\epsilon)}}{2}\bigg( u + a \ln \Big( r-r(1+\epsilon)^{-\frac{1}{2b}} \Big) \bigg)  + \frac{1-2\alpha-2\theta_{+}^{(n,\epsilon)}}{2}a \ln \Big( r(1-\epsilon)^{-\frac{1}{2b}}-r \Big) \\
& -\alpha u + \frac{a}{4}\bigg\{ \frac{1-a}{1-(1+\epsilon)^{-\frac{1}{2b}}} + \frac{1-a}{(1-\epsilon)^{-\frac{1}{2b}}-1} + (a-2b+4\alpha) \ln \bigg( \frac{(1-\epsilon)^{-\frac{1}{2b}}-1}{1-(1+\epsilon)^{-\frac{1}{2b}}} \bigg) \bigg\}, \\
& \widehat{C}_{3,4}^{(M)} = br^{2b} \bigg\{ u M^{2} + \int_{-M}^{M} (h_{1}(t)-2th_{0}(t))dt \bigg\}.
\end{align*}
It readily follows from \eqref{asymp of h0} that
\begin{align*}
\widehat{C}_{3,1}^{(M)} = \bigO\bigg( \frac{1}{M^{2}} \bigg), \qquad \widehat{C}_{3,2}^{(M)} = \bigO\bigg( \frac{1}{M^{2}} \bigg), \qquad \mbox{as } n \to + \infty.
\end{align*}
We also verify from \eqref{def of h0}, \eqref{def of h1}, \eqref{asymp of h0}, and 
\begin{equation*}
\frac{p_{1,a}(x)}{p_{0,a}(x)}=-\frac{a}{2}(1+2b)x-\frac{a}{2}\frac{a+2b(1-2a)}{x}+O(x^{-3}), \qquad \mbox{as } x \to \pm\infty 
\end{equation*}
that
\begin{align*}
h_{1}(t)-2th_{0}(t) = - \frac{a}{4b} \bigg( 1+2b+8b\ln \bigg( \frac{r|t|}{2b} \bigg) \bigg) t   -2ut \chi_{(0,\infty)}(t) + \frac{2ab-a^{2}}{4br^{2b}t} + \bigO(t^{-3})
\end{align*}
as $t \to \pm \infty$. We conclude that
\begin{align*}
\widehat{C}_{3,4}^{(M)} & = br^{2b}  \int_{-\infty}^{+\infty} \bigg\{ h_{1}(t)-2\,th_{0}(t) +  \frac{a}{4b} \bigg( 1+2b+8b\ln \bigg( \frac{r|t|}{2b} \bigg) \bigg) t + 2u t \chi_{(0,\infty)} (t) - \frac{(2ab-a^{2})t}{4br^{2b}(1+t^{2})}  \bigg\} \,dt  \\
& + \bigO(M^{-2}), \qquad \mbox{as } n \to + \infty,
\end{align*}
and the claim follows.
\end{proof}
\section{Proof of Theorem \ref{thm:main thm}}\label{section:proof thm}

By combining \eqref{log Dn as a sum of sums} with Lemmas \ref{lemma: S0}, \ref{lemma: S3}, \ref{lemma: S1} and \ref{lemma: asymp of S2k final}, we obtain
\begin{equation}
\ln \mathcal{E}_{n} = S_{0}+S_{1}+S_{2}+S_{3} =	C_{1} n + C_{2} \sqrt{n} + C_{3} +  \bigO \bigg( \frac{M^{3}}{\sqrt{n}}\ln n + \frac{\sqrt{n}}{M^{5}} + n^{-\frac{1}{2b}} \bigg),
\end{equation}
as $n \to +\infty$. Here we obtain the constants $C_{1}$, $C_{2}$ and $C_{3}$ of \eqref{asymp in main thm} after a long computation using \eqref{function F}--\eqref{function H}, a change of variables, and simplifying. Since $M = n^{\frac{1}{8}}(\ln n)^{-\frac{1}{8}}$, the error term is $\bigO\Big( n^{-\frac{1}{2b}} + (\ln n)^{\frac{5}{8}}n^{-\frac{1}{8}} \Big)$, which finishes the proof of Theorem \ref{thm:main thm}.

\appendix
\section{Uniform asymptotics of the incomplete gamma function}\label{section:uniform asymp gamma}
In this section, we collect some known asymptotic formulas for $\gamma(\tilde{a},z)$ that are useful for us. 
\begin{lemma}\label{lemma:various regime of gamma}(Taken from \cite[formula 8.11.2]{NIST}).
Let $\tilde{a}>0$ be fixed. As $z \to +\infty$,
\begin{align*}
\gamma(\tilde{a},z) = \Gamma(\tilde{a}) + \bigO(e^{-\frac{z}{2}}).
\end{align*}
\end{lemma}
\begin{lemma}\label{lemma: uniform}(Taken from \cite[Section 11.2.4]{Temme}).
For $\tilde{a}>0$ and $z>0$, we have
\begin{align*}
& \frac{\gamma(\tilde{a},z)}{\Gamma(\tilde{a})} = \frac{1}{2}\mathrm{erfc}(-\eta \sqrt{\tilde{a}/2}) - R_{\tilde{a}}(\eta), \qquad R_{\tilde{a}}(\eta) = \frac{e^{-\frac{1}{2}\tilde{a} \eta^{2}}}{2\pi i}\int_{-\infty}^{\infty}e^{-\frac{1}{2}\tilde{a} u^{2}}g(u)du,
\end{align*}
where $\mathrm{erfc}$ is defined in \eqref{def of erfc}, $g(u) = \frac{dt}{du}\frac{1}{\lambda-t}+\frac{1}{u+i \eta}$,
\begin{align}\label{lol8}
& \lambda = \frac{z}{\tilde{a}}, \quad \eta = (\lambda-1)\sqrt{\frac{2 (\lambda-1-\ln \lambda)}{(\lambda-1)^{2}}}, \quad u = -i(t-1)\sqrt{\frac{2(t-1-\ln t)}{(t-1)^{2}}},
\end{align}
where $\mathrm{sign} (\eta) = \mathrm{sign}(\lambda-1)$, and $\mathrm{sign}(u) = \mathrm{sign}(\im t)$ with $t \in \mathcal{L}:=\{\frac{\theta}{\sin \theta} e^{i\theta}: -\pi < \theta < \pi\}$ and $u \in \mathbb{R}$ (in particular $u = -i(t-1)+\bigO((t-1)^{2})$ as $t \to 1$). Furthermore, 
\begin{align}\label{asymp of Ra}
& R_{\tilde{a}}(\eta) \sim \frac{e^{-\frac{1}{2}\tilde{a} \eta^{2}}}{\sqrt{2\pi \tilde{a}}}\sum_{j=0}^{+\infty} \frac{c_{j}(\eta)}{\tilde{a}^{j}} \qquad \mbox{as } \tilde{a} \to + \infty
\end{align}
uniformly for $z \in [0,\infty)$, where all coefficients $c_{j}(\eta)$ are bounded functions of $\eta \in \mathbb{R}$ (i.e. bounded for $\lambda \in [0,\infty)$) and given by
\begin{align}\label{recursive def of the cj}
c_{0} = \frac{1}{\lambda-1} - \frac{1}{\eta}, \qquad c_{j} = \frac{1}{\eta} \frac{d}{d\eta}c_{j-1}(\eta) + \frac{\gamma_{j}}{\lambda-1}, \; j \geq 1,
\end{align}
where the $\gamma_{j}$ are the Stirling coefficients
\begin{align*}
\gamma_{j} = \frac{(-1)^{j}}{2^{j} \, j!} \bigg[ \frac{d^{2j}}{dx^{2j}} \bigg( \frac{1}{2}\frac{x^{2}}{x-\ln(1+x)} \bigg)^{j+\frac{1}{2}} \bigg]_{x=0}.
\end{align*}
The first few $c_{j}$ are $c_{0}(\eta) = \frac{1}{\lambda-1} - \frac{1}{\eta}$ and
\begin{align*}
& c_{1}(\eta) = \frac{1}{\eta^{3}}-\frac{1}{(\lambda-1)^{3}}-\frac{1}{(\lambda-1)^{2}}-\frac{1}{12(\lambda-1)}, \\
& c_{2}(\eta) = -\frac{3}{\eta^{5}} + \frac{3}{(\lambda-1)^{5}} + \frac{5}{(\lambda-1)^{4}} + \frac{25}{12(\lambda-1)^{3}} + \frac{1}{12(\lambda-1)^{2}} + \frac{1}{288(\lambda-1)}, \\
& c_{3}(\eta) \hspace{-0.035cm} = \hspace{-0.035cm} \frac{15}{\eta^{7}} \hspace{-0.04cm} - \hspace{-0.04cm} \frac{15}{(\lambda-1)^{7}} \hspace{-0.04cm} - \hspace{-0.04cm} \frac{35}{(\lambda-1)^{6}} \hspace{-0.04cm} - \hspace{-0.04cm} \frac{105}{4(\lambda-1)^{5}} \hspace{-0.04cm} - \hspace{-0.04cm} \frac{77}{12(\lambda-1)^{4}} \hspace{-0.04cm} - \hspace{-0.04cm} \frac{49}{288(\lambda-1)^{3}} \hspace{-0.04cm} - \hspace{-0.04cm} \frac{1}{288(\lambda-1)^{2}} \hspace{-0.04cm} + \hspace{-0.04cm} \frac{139}{51840(\lambda-1)}.
\end{align*}
In particular, the following hold:
\item[(i)] Let $\delta>0$ be fixed, and let $z=\lambda \tilde{a}$. As $\tilde{a} \to +\infty$, uniformly for $\lambda \geq 1+\delta$,
\begin{align*}
\gamma(\tilde{a},z) = \Gamma(\tilde{a})\big(1 + \bigO(e^{-\frac{\tilde{a} \eta^{2}}{2}})\big).
\end{align*}
\item[(ii)] Let $z=\lambda \tilde{a}$. As $\tilde{a} \to +\infty$, uniformly for $\lambda$ in compact subsets of $(0,1)$,
\begin{align*}
\gamma(\tilde{a},z) = \Gamma(\tilde{a})\bigO(e^{-\frac{\tilde{a} \eta^{2}}{2}}).
\end{align*}
\end{lemma} 

\paragraph{Acknowledgements.} SB acknowledges support from Samsung Science and Technology Foundation, Grant SSTF-BA1401-51, and the National Research Foundation of Korea, Grant NRF-2019R1A5A1028324. CC acknowledges support from the Novo Nordisk Fonden Project, Grant 0064428, the Swedish Research Council, Grant No. 2021-04626, and the Ruth and Nils-Erik Stenb\"ack Foundation.

\end{document}